\documentclass[11pt,a4paper]{article}
\usepackage{amsmath, amssymb, amsthm, mathtools}
\usepackage[margin=1in]{geometry}
\usepackage{enumerate}
\usepackage[T1]{fontenc}
\usepackage{hyperref}
\usepackage{xcolor}
\hypersetup{
  colorlinks=true,
  linkcolor=blue!50!black,
  citecolor=blue!50!black,
  urlcolor=blue!50!black
}

\theoremstyle{plain}
\newtheorem{theorem}{Theorem}[section]
\newtheorem{proposition}[theorem]{Proposition}
\newtheorem{lemma}[theorem]{Lemma}
\newtheorem{corollary}[theorem]{Corollary}

\theoremstyle{definition}
\newtheorem{definition}[theorem]{Definition}
\newtheorem{example}[theorem]{Example}

\theoremstyle{remark}
\newtheorem{remark}[theorem]{Remark}

\newcommand{\C}{\mathbb{C}}

\newcommand{\Z}{\mathbb{Z}}
\newcommand{\Q}{\mathbb{Q}}
\newcommand{\PGL}{\operatorname{PGL}}

\renewcommand{\d}{\,\mathrm{d}}

\title{\textbf{A Generalisation of Goursat's Algorithm for Integration in Finite Terms}}
\author{Sam Blake}
\date{\today}

\begin{document}
\maketitle

\begin{abstract}
We give a self-contained, modern exposition of \'Edouard Goursat's 1887 note on pseudo-elliptic integrals --- integrals $\int F(t)\,\d t/\sqrt{R(t)}$ with $R$ a cubic or quartic polynomial that, despite living on a genus-$1$ curve, admit elementary antiderivatives --- and we develop the analogous theory for $\int F(t)\,\d t/R(t)^{1/3}$ and $\int F(t)\,\d t/R(t)^{2/3}$ with $R$ cubic or quadratic. The mechanism is the same in every case: a M\"obius symmetry of the branch points lifts to automorphisms of the curve; a differential that is invariant under a lift \emph{with fixed points} is the pull-back of a rational differential and so integrates elementarily, while a differential invariant only under \emph{translations} descends to an isogenous curve, where nothing is gained. For cube roots, the order-$3$ M\"obius map cycling the roots of $R$ splits the integrand into three eigenpieces; two are always elementary, and so is half of the third, by the involutions exchanging two roots. These criteria are sufficient, not necessary: the remaining piece is an instance of the original problem on an isomorphic curve. We therefore complement them with Hermite reduction, which removes exact parts and yields a rigorous certificate of non-elementarity whenever the residual differential is of the second kind. Finally we restate everything in terms of covariants of $R$ --- the resolvent cubic for square roots, the Hessian and the Jacobian covariant for cube roots --- so that neither the roots of $R$, nor roots of unity, nor fixed points are needed in computation.
\end{abstract}

\section{Introduction}
\label{sec:intro}

Among the very first surprises of integral calculus is the asymmetry between
\[
  \int \frac{\d t}{\sqrt{1-t^2}} = \arcsin(t)
\qquad\text{and}\qquad
  \int \frac{\d t}{\sqrt{(1-t^2)(1-k^2t^2)}} = ?
\]
The first integral is elementary; the second is not. Both integrands look formally similar --- a rational function divided by the square root of a polynomial --- yet only the first admits a closed-form antiderivative in terms of the standard library of elementary functions. The deep reason is geometric: the curve $y^2 = 1 - t^2$ has \emph{genus zero} (it is the unit circle), while the curve $y^2 = (1-t^2)(1-k^2 t^2)$ has \emph{genus one}, and on a curve of positive genus most $1$-forms have no elementary antiderivative. The classical incomplete elliptic integral on the right is the prototype.\\

This raises a refined question: can a $1$-form on a positive-genus curve \emph{ever} integrate elementarily? The answer is yes, but only for very special integrands. The \emph{pseudo-elliptic integrals} are precisely those of the form $\int F(t)\,\d t/\sqrt{R(t)}$ which look elliptic but evaluate elementarily. Their existence is non-obvious, and producing systematic methods for detecting and reducing them was a recurrent theme in nineteenth-century analysis. \'Edouard Goursat's note of 1887~\cite{Goursat1887} is the most elegant general account.\\

The purpose of this paper is threefold. First, in Sections~\ref{sec:liouville}--\ref{sec:goursat}, we give a self-contained modern exposition of Goursat's note, including his substitutions of order three (which are often omitted from accounts of his work). Second, in Sections~\ref{sec:cuberoot}--\ref{sec:examples}, we develop a cube-root analog. The natural setting is the curve $y^3 = R(t)$, and the binary anti-invariance condition $F(t)+F(S(t))=0$ is replaced by an eigenvalue condition under an order-$3$ M\"obius substitution, refined by the involutions that exchange two roots of $R$. Third, we are careful about what such symmetry criteria can and cannot prove. They are \emph{sufficient} conditions for elementarity. When they fail, the integral may still be elementary; to decide this one needs different tools, and we show how far Hermite reduction goes.\\

Two classical results frame the cube-root theory and should be kept in mind. The first is Chebyshev's theorem on binomial differentials~\cite{Chebyshev1853}: $\int t^m(a+bt^n)^{p}\,\d t$ with rational $m,n,p$ is elementary if and only if one of $p$, $(m+1)/n$, $(m+1)/n+p$ is an integer. For $R=t^n-K$ and monomial $F$ our eigenspace analysis reproduces exactly this trichotomy; what is added is the passage to arbitrary rational $F$ and, through a M\"obius normalisation, to arbitrary $R$. The second is Goursat's own treatment of order-$3$ substitutions on $y^2=t^3-1$ (\S\S4--5 of~\cite{Goursat1887}); since every genus-$1$ curve $y^3=R(t)$ is birational to a curve $Y^2=X^3+c$, the cube-root theory is the same geometry in a different coordinate (Section~\ref{sec:relations}).\\

The possibility of a cube-root extension of Goursat's procedure was hinted at over many years by Martin Welz in postings on \texttt{sci.math.symbolic}~\cite{WelzNewsgroup}, where particular cube-root integrands were observed to behave in a Goursat-like manner under cyclic M\"obius substitutions.

\medskip

The structure of the paper is as follows. Section~\ref{sec:liouville} reviews integration in finite terms, Liouville's principle, and Hermite reduction on the curves $y^n=R$. Section~\ref{sec:goursat} presents Goursat's theory: Theorem~\ref{thm:goursat1} (single involution), Theorem~\ref{thm:goursat2} (Klein four-group), Theorem~\ref{thm:order3} (order-$3$ substitutions), and a root-free formulation. Section~\ref{sec:cuberoot} develops the cube-root analog: Theorem~\ref{thm:cubic-goursat}, the exponent $2/3$, the transpositions, the formulation through the Hessian (Proposition~\ref{prop:invariant}), and the identification of the residual with a Goursat integral (Proposition~\ref{prop:second-level}). Section~\ref{sec:examples} works through examples, including several in which the symmetry criterion fails although the integral is elementary. Section~\ref{sec:remarks} closes with remarks on higher radicals and on implementation.

\section{Integration in finite terms and Liouville's theorem}
\label{sec:liouville}

The framework for making ``elementary integration'' precise is the differential algebra of Joseph Liouville, refined by Ritt~\cite{Ritt1948} and Rosenlicht~\cite{Rosenlicht1972}, and given its modern algorithmic formulation through the work of Risch~\cite{Risch1969,Risch1970}, Davenport~\cite{Davenport1981}, Trager~\cite{Trager1984}, and Bronstein~\cite{Bronstein2005}.

\begin{definition}\label{def:diffield}
A \emph{differential field} is a field $K$ of characteristic zero equipped with a derivation, i.e.\ a map $D \colon K \to K$ satisfying $D(a+b)=Da+Db$ and $D(ab)=aDb+bDa$. We write $a' = Da$ throughout.
\end{definition}

The basic example is $K = \C(t)$ with $D = \d/\d t$. We will also consider algebraic extensions $K\subseteq L$, where the derivation extends uniquely.

\begin{definition}\label{def:elementary}
Let $K\subseteq L$ be differential fields. We say $L$ is an \emph{elementary extension} of $K$ if there is a tower
\[
  K = K_0 \subseteq K_1 \subseteq K_2 \subseteq \cdots \subseteq K_n = L
\]
such that each $K_{i+1} = K_i(\theta_i)$ where one of the following holds:
\begin{enumerate}[(i)]
  \item $\theta_i$ is \emph{algebraic} over $K_i$;
  \item $\theta_i$ is an \emph{exponential}: $\theta_i'/\theta_i \in K_i$;
  \item $\theta_i$ is a \emph{logarithm}: $\theta_i' = \eta'/\eta$ for some $\eta \in K_i$.
\end{enumerate}
An element $f \in K$ is said to have an \emph{elementary integral} (or \emph{integrate in finite terms}) if there exists an elementary extension $L \supseteq K$ and an element $g\in L$ with $g' = f$.
\end{definition}

The class of elementary functions in the calculus textbook sense corresponds to elementary extensions of $\C(t)$.

\begin{theorem}[Liouville's Theorem; \cite{Rosenlicht1972}]\label{thm:liouville}
Let $K$ be a differential field with field of constants $\C$. If $f\in K$ has an elementary integral, then there exist $v\in K$, constants $c_1,\dots,c_n\in \C$, and elements $u_1,\dots,u_n \in K$ such that
\begin{equation}\label{eq:liouville}
  f \;=\; v' \;+\; \sum_{i=1}^n c_i\,\frac{u_i'}{u_i}.
\end{equation}
\end{theorem}

The remarkable feature of \eqref{eq:liouville} is that no new transcendentals appear: although the antiderivative may live in an extension $L\supsetneq K$, the structural form of $f$ is dictated by elements of $K$ itself. This is the precise statement that ``the integral, if elementary, is the derivative of a rational function plus a sum of constant multiples of logarithms of rational functions.''\\

For our purposes we will need to apply Liouville's theorem when $K$ is an algebraic --- specifically, a finite radical --- extension of a rational function field. Throughout the rest of the paper we restrict attention to this purely algebraic setting, ignoring any further transcendental extensions: the integrands of interest are of the form $F(t)/\sqrt{R(t)}$ or $F(t)/\sqrt[3]{R(t)}$, and the relevant differential field is
\[
  K = \C(t,\theta), \qquad \theta^n = R(t),\quad n\in\{2,3\}.
\]

\begin{theorem}[Liouville's Theorem in the algebraic case]\label{thm:liouville-alg}
Let $K = \C(t)(\theta)$ where $\theta$ is algebraic over $\C(t)$, and let $f\in K$. If $\int f\,\d t$ is elementary, then it can be written
\[
  \int f\,\d t = v + \sum_{i=1}^m c_i \log u_i
\]
for some $v, u_1,\dots,u_m \in K$ and constants $c_i \in \C$.
\end{theorem}

This is the form of Liouville's principle relevant to the integration of algebraic functions in the present paper. We shall use it through the following standard consequence. Recall that a meromorphic $1$-form on a smooth projective curve $X$ is of the \emph{first kind} if it is holomorphic, of the \emph{second kind} if all its residues vanish, and of the \emph{third kind} otherwise.

\begin{lemma}\label{lem:second-kind}
Let $X$ be a smooth projective curve over $\C$ and $\omega$ a differential of the second kind on $X$. Then $\int\omega$ is elementary if and only if $\omega$ is exact, i.e.\ $\omega=\d f$ for some $f\in\C(X)$.
\end{lemma}

\begin{proof}
If $\int\omega$ is elementary then by Theorem~\ref{thm:liouville-alg} $\omega=\d v+\sum_i c_i\,\d u_i/u_i$ with $v,u_i\in\C(X)$. Choose a $\Q$-basis $e_1,\dots,e_s$ of the $\Q$-span of the $c_i$ and write $c_i=\sum_j q_{ij}e_j$ with $q_{ij}\in\Z$ (after rescaling the $e_j$). Then $\sum_ic_i\,\d u_i/u_i=\sum_je_j\,\d w_j/w_j$ with $w_j=\prod_iu_i^{q_{ij}}$. The residue of the right-hand side at a point $P$ is $\sum_je_j\operatorname{ord}_P(w_j)$; since $\d v$ has no residues and $\omega$ has none by hypothesis, this vanishes, and by $\Q$-independence $\operatorname{ord}_P(w_j)=0$ for all $j$ and $P$. So every $w_j$ is constant and $\omega=\d v$.
\end{proof}

In particular a non-zero differential of the first kind is never elementary: if $\omega=\d f$ is holomorphic then $f$ has no poles, so $f$ is constant. For differentials of the second kind the situation is different, and this point is easily overlooked: \emph{a non-zero differential of the second kind may perfectly well be exact} --- every $\d f$ is one. The space of second-kind differentials modulo exact ones is $H^1_{\mathrm{dR}}(X)$, of dimension $2g$; to show that a given second-kind differential is not elementary one has to show that its \emph{class} is non-zero. On the curves of this paper that is a finite computation.

\subsection*{Hermite reduction on $y^n=R$}

Let $R\in\C[t]$ be squarefree of degree $d$, let $n\in\{2,3\}$, $1\le k<n$, and consider $F\,\d t/y^k$ on $y^n=R$, with $F\in\C(t)$. Put $e=1-k/n$. Since $y'=R'y/(nR)$,
\begin{equation}\label{eq:exact-forms}
  \d\bigl(A\,y^{n-k}\bigr)\;=\;\bigl(A'R+e\,A\,R'\bigr)\frac{\d t}{y^k}\qquad(A\in\C(t)).
\end{equation}

\begin{proposition}[Hermite reduction]\label{prop:hermite}
Every $F\in\C(t)$ can be written uniquely as
\begin{equation}\label{eq:hermite}
  F \;=\; A'R+e\,A\,R' \;+\; \widetilde F,\qquad \widetilde F \;=\; P+\frac{N}{D},
\end{equation}
with $A\in\C(t)$, $P,N,D\in\C[t]$, $\deg P\le d-2$, $D$ squarefree and coprime to $R$, and $\deg N<\deg D$. The differential $F\,\d t/y^k$ is exact if and only if $\widetilde F=0$. Moreover:
\begin{enumerate}[(a)]
  \item if $N\neq 0$, then $\widetilde F\,\d t/y^k$ is of the third kind;
  \item if $N=0$, then $\widetilde F\,\d t/y^k$ is of the second kind if and only if $n\nmid dk$, or $n\mid dk$ and the coefficient of $t^{-1}$ in the expansion of $P\,R^{-k/n}$ at $t=\infty$ vanishes.
\end{enumerate}
Consequently, if $\widetilde F\ne0$ and $\widetilde F\,\d t/y^k$ is of the second kind, then $\int F\,\d t/y^k$ is \emph{not} elementary.
\end{proposition}

\begin{proof}
\emph{Existence.} Let $s$ be a pole of $F$ of order $m$. If $R(s)\neq0$ and $m\ge2$, then $A=c\,(t-s)^{-(m-1)}$ contributes $-(m-1)\,c\,R(s)\,(t-s)^{-m}+\cdots$ to $A'R+eAR'$, so a suitable $c$ lowers the order of the pole. If $R(s)=0$ and $m\ge1$, then $A=c\,(t-s)^{-m}$ contributes $c\,(e-m)\,R'(s)\,(t-s)^{-m}+\cdots$, and $e-m\neq0$ because $e\notin\Z$; so poles at the roots of $R$ can be removed entirely. Finally $A=c\,t^j$ with $j\ge0$ contributes $c\,(j+ed)\,a\,t^{j+d-1}+\cdots$, where $a$ is the leading coefficient of $R$ and $j+ed>0$, which lowers the degree of the polynomial part as long as it is at least $d-1$.

\emph{Uniqueness.} Suppose $A'R+eAR'$ has the shape of $\widetilde F$. By the local computations above, a pole of $A$ away from the roots of $R$ produces a pole of order at least $2$, and a pole of $A$ at a root of $R$ produces a pole there; so $A$ is a polynomial, and if $A\neq0$ has degree $j$ then $A'R+eAR'$ has degree $j+d-1>d-2$. Hence $A=0$. This also proves the exactness criterion: by~\eqref{eq:exact-forms} and the fact that the functions on the curve with $\rho^*$-eigenvalue matching that of $\d t/y^k$ (where $\rho\colon y\mapsto\zeta_ny$) are exactly the $A\,y^{n-k}$, the differential $F\,\d t/y^k$ is exact iff $F=A'R+eAR'$ for some $A$.

\emph{(a)} A zero $s$ of $D$ is not a branch point, so there are $n$ points above it, at each of which $\widetilde F\,\d t/y^k$ has a simple pole with residue $N(s)/(D'(s)y^k)\neq0$.
\emph{(b)} If $N=0$ the only possible poles lie over $t=\infty$. If $n\nmid d$ there is a single point there, and the residue vanishes because the sum of all residues does; in this case $n\nmid dk$ as $n$ is prime. If $n\mid d$ there are $n$ points over $\infty$, at which $1/t$ is a local parameter, and the residues are the stated coefficient up to sign and a root of unity.

The last statement follows from Lemma~\ref{lem:second-kind}.
\end{proof}

Equation~\eqref{eq:hermite} is a linear system for the coefficients of $A$, $P$, $N$, so the reduction is one linear solve. On a genus-$1$ curve the reduced polynomial part is small. For example, for $n=3$, $k=1$ and $R$ cubic, $\widetilde F$ is spanned by $1$ and $t$ once the poles are gone: $\d t/y$ is of the third kind (simple poles at the three points over $\infty$), while $\bigl(t+\tfrac{b}{3a}\bigr)\d t/y$ is of the second kind and not exact. For $n=2$ and $R$ cubic, $\widetilde F\in\langle1,t\rangle$: $\d t/y$ is of the first kind and $t\,\d t/y$ of the second.

When $\widetilde F\,\d t/y^k$ is of the third kind, elementarity is governed by whether certain divisors supported on its poles are torsion in the Jacobian (see~\cite{Davenport1981,Trager1984}); the symmetry methods of this paper detect some, but not all, such cases.

\begin{remark}
The constructive content of Theorem~\ref{thm:liouville-alg} is the algebraic case of the Risch--Trager--Bronstein algorithm~\cite{Risch1969,Risch1970,Davenport1981,Trager1984,Bronstein2005}: given $f\in K$ as input, the algorithm decides in finite time whether $\int f\,\d t$ is elementary, and if so produces it. Goursat's theorems below predate these developments by nearly a century and provide closed-form reductions in cases where the algebraic structure is particularly transparent.
\end{remark}

\section{Goursat's pseudo-elliptic theorem}
\label{sec:goursat}

Throughout this section, $R(t)\in\C[t]$ denotes a polynomial of degree $3$ or $4$ with \emph{simple roots} — i.e., all roots distinct, or equivalently $\gcd(R,R') = 1$. We write $C := \{(t,y)\in\overline{\C^2} : y^2 = R(t)\}$ for the smooth projective completion of the affine curve $y^2=R(t)$; this is a curve of genus~$1$. We are interested in the integrals
\begin{equation}\label{eq:elliptic-integrand}
  I = \int F(t)\,\frac{\d t}{\sqrt{R(t)}}
\end{equation}
with $F\in\C(t)$.

\begin{definition}\label{def:pseudo-elliptic}
The integral~\eqref{eq:elliptic-integrand} is called \emph{pseudo-elliptic} if it is elementary in the sense of Definition~\ref{def:elementary}, i.e.\ if it admits a closed-form antiderivative built from rational functions, $\sqrt{R}$, exponentials, and logarithms.
\end{definition}

The reason for the term \emph{pseudo-elliptic} is that the integrand lives most naturally as a $1$-form on $C$, and a generic such integral is genuinely transcendental (an elliptic integral). The pseudo-elliptic ones are the rare exceptions where the integral happens to be elementary.

\subsection{The role of M\"obius involutions}

The key idea is to exploit symmetries of the curve $C$. The M\"obius group $\PGL_2(\C)$ acts on $\C\cup\{\infty\}=\mathbb{P}^1(\C)$ by fractional linear transformations $S(t) = (At+B)/(Ct+D)$. An element $S \in \PGL_2(\C)$ is an \emph{involution} if $S^2 = \mathrm{id}$ and $S\neq \mathrm{id}$. A non-trivial M\"obius involution has exactly two fixed points on $\mathbb{P}^1(\C)$.

\begin{lemma}\label{lem:involutions-pairs}
Let $a,b,c,d\in\mathbb{P}^1(\C)$ be distinct. There is a unique M\"obius transformation $S$ with $S(a)=b$, $S(b)=a$, $S(c)=d$; it is an involution, it satisfies $S(d)=c$, and for finite $a,b,c,d$
\begin{equation}\label{eq:involution-formula}
  S(t) \;=\; \frac{(ab-cd)\,t + (a+b)cd - (c+d)ab}{[(a+b)-(c+d)]\,t - (ab-cd)}.
\end{equation}
There are exactly three such involutions, corresponding to the three ways of pairing four points; together with the identity they form a Klein four-group $V_4 \subseteq \PGL_2(\C)$.
\end{lemma}

\begin{proof}
A M\"obius transformation is determined by its action on three points, so $S$ exists and is unique. A M\"obius transformation exchanging two points is an involution (conjugate so that the points are $0,\infty$; then $S(z)=\kappa/z$). Hence $S(d)=S(S(c))=c$, and one verifies~\eqref{eq:involution-formula} by direct computation. The product of any two of the three involutions acts on $\{a,b,c,d\}$ as the third, so they generate a $V_4$ subgroup.
\end{proof}

\begin{remark}\label{rmk:other-involutions}
An involution permuting four points need not be of this type: it may fix two of them and exchange the other two. This happens only for special configurations (cross-ratio $-1$), for instance $S(t)=1/t$ on the roots of $t^4-1$. Such involutions turn out to give nothing new (Section~\ref{sec:order3}).
\end{remark}

\subsection{Goursat's first theorem}

\begin{theorem}[Goursat~\cite{Goursat1887}, Theorem 1]\label{thm:goursat1}
Let $R\in\C[t]$ be a polynomial of degree $3$ or $4$ with simple roots, and let $S\in\PGL_2(\C)$ be a non-trivial M\"obius involution permuting the roots of $R$ in pairs. If $F\in\C(t)$ satisfies the \emph{anti-invariance} condition
\begin{equation}\label{eq:antiinv}
  F(t) + F(S(t)) = 0\qquad \text{(identically)},
\end{equation}
then the integral $\int F(t)\,\d t/\sqrt{R(t)}$ is pseudo-elliptic. Explicitly, let $\alpha,\beta\in\mathbb{P}^1(\C)$ be the two fixed points of $S$, and set
\[
  u = \frac{t-\alpha}{t-\beta},\qquad x = u^2.
\]
Then
\[
  \int F(t)\,\frac{\d t}{\sqrt{R(t)}} \;=\; \int \frac{G(x)\,\d x}{2\sqrt{Q(x)}}
\]
for some $G\in\C(x)$ and quadratic $Q\in\C[x]$, the latter integral being elementary.
\end{theorem}

\begin{proof}
We argue in five steps.

\smallskip\noindent
\textbf{Step 1 (canonical form of $S$).}
Since $S$ has order $2$ with two distinct fixed points $\alpha,\beta$, the M\"obius transformation $u = (t-\alpha)/(t-\beta)$ sends $\alpha\mapsto 0$, $\beta\mapsto\infty$, and conjugates $S$ into a M\"obius transformation of $u$ fixing $\{0,\infty\}$. Such a transformation is multiplication by a scalar $\lambda$, and the involution condition $\lambda^2=1$ together with $\lambda\neq 1$ forces $\lambda = -1$. Thus in $u$-coordinates, $S$ acts as $u\mapsto -u$. We write $t(u) = (\alpha-\beta u)/(1-u)$ for the inverse change of variable.

\smallskip\noindent
\textbf{Step 2 (transformation of $R$).}
Each linear factor $t-r_i$ becomes
\[
  t(u) - r_i \;=\; \frac{(\alpha-r_i) - u(\beta-r_i)}{1-u}.
\]

Hence (assuming $R$ is quartic; in the cubic case one root is $r_4=\infty$, the corresponding factor is $1-u$, and the argument is unchanged)

\[
  R(t(u))(1-u)^4 \;=\; \prod_{i=1}^4 \big[(\alpha-r_i) - u(\beta-r_i)\big].
\]

The roots of this polynomial in $u$ are the images $u_i=(r_i-\alpha)/(r_i-\beta)$ of the $r_i$. Since $S$ swaps the roots $r_i$ in pairs and acts as $u\mapsto -u$, the $u_i$ come in pairs $\{u_1,-u_1\}$ and $\{u_3,-u_3\}$. Therefore

\begin{equation}\label{eq:R-becomes-even}
  R(t(u))(1-u)^4 \;=\; C\,(u^2 - u_1^2)(u^2 - u_3^2)
\end{equation}
for some non-zero constant $C$. The right-hand side is a polynomial in $u^2$.

\smallskip\noindent
\textbf{Step 3 (transformation of the differential).}
A direct computation gives
\[
  \d t = \frac{\alpha-\beta}{(1-u)^2}\,\d u.
\]
Combined with~\eqref{eq:R-becomes-even} this yields
\begin{equation}\label{eq:differential-becomes}
  \frac{\d t}{\sqrt{R(t)}} \;=\; \frac{(\alpha-\beta)\,\d u/(1-u)^2}{\sqrt{C\,(u^2-u_1^2)(u^2-u_3^2)}/(1-u)^2}
  \;=\; \frac{(\alpha-\beta)\,\d u}{\sqrt{C\,(u^2-u_1^2)(u^2-u_3^2)}}.
\end{equation}
The factors $(1-u)^2$ from the Jacobian and from the radical exactly cancel.

\smallskip\noindent
\textbf{Step 4 (oddness of $F$ in $u$).}
Anti-invariance~\eqref{eq:antiinv} reads, in $u$-coordinates, $F(t(-u)) = F(S(t(u))) = -F(t(u))$. Hence $u\mapsto F(t(u))$ is an odd rational function of $u$, and admits a representation
\[
  F(t(u)) \;=\; u\,G(u^2)
\]
for some $G\in\C(x)$.

\smallskip\noindent
\textbf{Step 5 (descent via $x = u^2$).}
Combining the previous steps:
\[
  \int F(t)\,\frac{\d t}{\sqrt{R(t)}}
  \;=\; (\alpha-\beta)\,\int \frac{u\,G(u^2)\,\d u}{\sqrt{C\,(u^2-u_1^2)(u^2-u_3^2)}}.
\]
Substituting $x=u^2$, $\d x = 2u\,\d u$:
\[
  \;=\; \frac{\alpha-\beta}{2\sqrt{C}}\,\int \frac{G(x)\,\d x}{\sqrt{(x-u_1^2)(x-u_3^2)}}.
\]
The radical is now $\sqrt{Q(x)}$ for the quadratic $Q(x)=(x-u_1^2)(x-u_3^2)$. The integral $\int G(x)\,\d x/\sqrt{Q(x)}$ with $G$ rational and $Q$ quadratic is elementary by classical methods (partial fractions of $G$ together with Euler substitution). This completes the proof.
\end{proof}

\begin{remark}
The case where one of the fixed points is at infinity (so that $S(t) = 2\alpha - t$) is covered by the same argument, the substitution simplifying to $u = t-\alpha$.
\end{remark}

\subsection{Goursat's second theorem}

The three involutions of Lemma~\ref{lem:involutions-pairs} together with the identity form $V_4$, and one can ask for a sum-of-orbits criterion that detects pseudo-ellipticity.

\begin{theorem}[Goursat~\cite{Goursat1887}, Theorem 2]\label{thm:goursat2}
Let $R\in\C[t]$ be a polynomial of degree $3$ or $4$ with simple roots, with M\"obius involutions $S_1,S_2,S_3$ permuting the roots in pairs, and let $F\in\C(t)$. If
\begin{equation}\label{eq:V4-sum}
  F + F\circ S_1 + F\circ S_2 + F\circ S_3 \;=\; 0
\end{equation}
identically, then $\int F\,\d t/\sqrt{R}$ is pseudo-elliptic.
\end{theorem}

\begin{proof}
The group $V_4=\{I,S_1,S_2,S_3\}$ has four irreducible characters $\chi_0,\chi_1,\chi_2,\chi_3 \colon V_4\to\{\pm 1\}$, with $\chi_0$ the trivial character. We label them so that $\chi_j(S_j) = +1$ for $j=1,2,3$, and $\chi_j(S_k) = -1$ for $j\neq k$ in $\{1,2,3\}$.

The associated character projections decompose $\C(t)$ as a direct sum of $\chi_j$-eigenspaces under the natural $V_4$-action. Setting
\[
  F^{(j)} \;:=\; \frac{1}{4}\sum_{g\in V_4} \chi_j(g)\,F\circ g,
\]
hypothesis~\eqref{eq:V4-sum} says exactly $F^{(0)} = 0$, hence $F = F^{(1)} + F^{(2)} + F^{(3)}$.

We claim $F^{(j)}$ is anti-invariant under $S_k$ for every $k\neq j$ (and invariant under $S_j$, as is automatic). Indeed, applying $S_k$ to $F^{(j)}$ permutes the summands by the action $g\mapsto gS_k$, which on characters acts by $\chi_j(g)\mapsto \chi_j(gS_k) = \chi_j(g)\chi_j(S_k)$. For $k\neq j$ this multiplies the projection by $-1$, giving anti-invariance.

By Theorem~\ref{thm:goursat1}, applied to $F^{(j)}$ with any involution $S_k$, $k\neq j$, each integral $\int F^{(j)}\,\d t/\sqrt{R}$ is pseudo-elliptic. Summing the three antiderivatives yields an elementary expression for $\int F\,\d t/\sqrt{R}$.
\end{proof}

\begin{remark}\label{rmk:labelling}
With the labelling above, $F^{(j)}$ is invariant under $S_j$ and anti-invariant under the other two involutions. Goursat's own indexing is different but equivalent: his $f=\tfrac14\bigl(F-F\circ S_1+F\circ S_2-F\circ S_3\bigr)$ is our $F^{(2)}$ and satisfies $f\circ S_1+f=0$. When implementing the theorem one must take care to pair each $F^{(j)}$ with an involution $S_k$, $k\neq j$.
\end{remark}

\begin{remark}[The criterion is not necessary]\label{rmk:not-necessary}
Condition~\eqref{eq:V4-sum} is sufficient for $\int F\,\d t/\sqrt R$ to be elementary; it is not necessary. For $R=(t^2-1)(t^2-4)$ and $F=3t^4-10t^2+4$ one has $\int F\,\d t/\sqrt R=t\sqrt R$, yet $F^{(0)}=(3t^8-10t^6+8t^4-40t^2+48)/(2t^4)\neq0$. The reason is that the projections commute with $\d$: the components of an exact differential are exact, and nothing forces $F^{(0)}$ to vanish. Applying Hermite reduction (Proposition~\ref{prop:hermite}) \emph{before} the symmetry analysis removes this defect, and applying it to the residual $F^{(0)}$ afterwards decides the question whenever the result is of the second kind.
\end{remark}

\subsection{Geometric interpretation}\label{sec:geometric}

Although Goursat does not phrase his theorems in this way, the modern reading explains all of them at once, and it will carry over verbatim to cube roots. Let $X$ be a curve of genus $1$. Every automorphism of $X$ is either a \emph{translation} (and then has no fixed points, and preserves the holomorphic differential) or \emph{has fixed points} (and then the quotient of $X$ by the group it generates has genus $0$).

\begin{quote}\itshape
A differential on $X$ that is invariant under an automorphism with fixed points is the pull-back of a differential on $\mathbb{P}^1$, hence integrates elementarily. A differential that is invariant only under translations is the pull-back of a differential on an isogenous curve of genus $1$, and nothing has been gained.
\end{quote}

A M\"obius involution $S$ pairing the roots of $R$ has two lifts to $C\colon y^2=R(t)$, of the form $(t,y)\mapsto(S(t),\pm y\,J_S(t))$ and differing by the hyperelliptic involution. One of them, $\widetilde S_+$, is the translation by a point of order two; under it $\d t/y$ is invariant. The other, $\widetilde S_-$, has four fixed points (over the two fixed points of $S$); under it $\d t/y$ changes sign. Hence $F\,\d t/y$ is $\widetilde S_-$-invariant precisely when $F\circ S=-F$, and Theorem~\ref{thm:goursat1} is the statement that such a differential descends to $C/\widetilde S_-\cong\mathbb{P}^1$; the conic $v^2=Q(x)$ of the theorem is that quotient. The $S$-\emph{invariant} part of $F$, on the other hand, gives a differential invariant under the translation $\widetilde S_+$, which descends to the genus-$1$ curve $v^2=x\,Q(x)$.

In Theorem~\ref{thm:goursat2} three involutions are used in turn. What is left over, $F^{(0)}\,\d t/y$, is invariant under all three translations, i.e.\ under the group $C[2]$ of points of order two. It is therefore the pull-back of a differential under $C\to C/C[2]$, and that map is multiplication by $2$ on $C$ itself: $C/C[2]\cong C$. \emph{The obstruction to Goursat's method is an instance of the original problem on the same curve.} This is why the criterion cannot be made necessary by iterating it.

\subsection{A root-free formulation}\label{sec:rootfree-sqrt}

Neither the roots of $R$ nor the fixed points of the involutions are needed to apply Theorems~\ref{thm:goursat1} and~\ref{thm:goursat2}. Write a quartic as $R=a\,(t^4-\sigma_1t^3+\sigma_2t^2-\sigma_3t+\sigma_4)$, so that the $\sigma_i$ are the elementary symmetric functions of the roots.

\begin{proposition}\label{prop:rootfree-sqrt}
\begin{enumerate}[(a)]
\item The three involutions of Lemma~\ref{lem:involutions-pairs} are $S_\theta(t)=(Nt-M)/(Lt-N)$ with
\[
  (L:N:M)\;=\;\bigl(\sigma_1^2-4\sigma_2+4\theta\;:\;\sigma_1\theta-2\sigma_3\;:\;2\theta(\sigma_2-\theta)-\sigma_1\sigma_3\bigr),
\]
where $\theta=r_ar_b+r_cr_d$ runs over the roots of the resolvent cubic
$\theta^3-\sigma_2\theta^2+(\sigma_1\sigma_3-4\sigma_4)\theta-(\sigma_1^2\sigma_4+\sigma_3^2-4\sigma_2\sigma_4)$.
(If this triple vanishes at a root, use $(\sigma_1\theta-2\sigma_3:\theta^2-4\sigma_4:\theta\sigma_3-2\sigma_1\sigma_4)$, which is proportional to it.) For a cubic $R=a(t^3-\sigma_1t^2+\sigma_2t-\sigma_3)$ the involutions are indexed by the roots $r$ of $R$ itself, with $(L:N:M)=(1:r:2r\sigma_1-2r^2-\sigma_2)$.
\item \emph{(Goursat's criterion over the ground field.)} $F^{(0)}=\tfrac14\bigl(F+\sum_\theta F\circ S_\theta\bigr)$ is a trace over the cubic, computable without solving it. As Goursat remarks, ``the sum is a symmetric function of the roots; one can recognise whether it vanishes without knowing the roots of the polynomial under the radical''.
\item \emph{(The reduction without fixed points.)} Let $S=S_\theta$ and $X(t)=t+S(t)=(Lt^2-M)/(Lt-N)$, a function of degree $2$ generating the field of $S$-invariant functions. Then $R(t)=Q(X)\,(Lt-N)^2$ for a polynomial $Q$ of degree at most $2$, and if $F\circ S=-F$,
\[
  \int F\,\frac{\d t}{\sqrt R}\;=\;\int G(X)\,\frac{\d X}{\sqrt{Q(X)}},\qquad G(X(t))=\frac{F(t)}{X'(t)\,(Lt-N)} .
\]
\item If $R=q_1q_2$ is the factorisation into quadratics corresponding to the pairing, then $m=\sqrt R/q_1$ is a rational parameter on the conic $v^2=Q(x)$, so the integral in~(c) becomes $\int\rho(m)\,\d m$ with $\rho$ rational and $m$ expressed through $\sqrt{R(t)}$ itself.
\end{enumerate}
\end{proposition}

\begin{proof}
(a) Write $R/a=q_1q_2$ with $q_i=t^2-s_it+p_i$. By~\eqref{eq:involution-formula}, $(L:N:M)=(s_1-s_2:p_1-p_2:p_1s_2-p_2s_1)$, which changes sign when $q_1,q_2$ are exchanged; multiplying by $s_1-s_2$ gives expressions symmetric in the two factors, hence polynomial in $\theta=p_1+p_2$ and the $\sigma_i$, namely those displayed. (b) is clear. (c) is \S1 of Goursat's note read backwards: $X$ is $S$-invariant of degree $2$, a quartic whose roots are permuted in pairs by $S$ has the form $A\,P^2+B\,PD+C\,D^2$ with $X=P/D$, and $F/(X'(Lt-N))$ is a quotient of two anti-invariant functions. (d) The two roots of $Q$ are the values of $X$ on the two pairs of roots of $R$.
\end{proof}

The function $G$ in (c) is found by two polynomial divisions: if $g=g_1/g_2$ is $S$-invariant then $g_1\equiv G(x)\,g_2$ modulo $(Lt^2-M)-x\,(Lt-N)$ in $\C(x)[t]$, so $G$ is the ratio of corresponding coefficients of the two remainders. Part~(d) also takes care of a point that is easy to get wrong: if the answer is written in terms of $\sqrt{Q(X(t))}$ it agrees with $\sqrt{R(t)}$ only up to a sign depending on $t$, whereas an answer written in terms of $m=\sqrt{R}/q_1$ differentiates to the integrand identically.

\subsection{Goursat's substitutions of order three and four}\label{sec:order3}

For special $R$ the four branch points admit more M\"obius symmetries than $V_4$. The full group is the alternating group $A_4$ when the branch points are equianharmonic ($j=0$; for a quartic $at^4+bt^3+ct^2+dt+e$ this is the vanishing of the invariant $I=12ae-3bd+c^2$, and for a cubic $at^3+bt^2+ct+d$ it is $b^2=3ac$), and the dihedral group of order $8$ when they are harmonic ($j=1728$). Goursat treats both in \S\S4--6 of his note.

\begin{theorem}[Goursat~\cite{Goursat1887}, \S\S4--5]\label{thm:order3}
Suppose a M\"obius transformation $g$ of order $3$ fixes one branch point $r$ of $y^2=R(t)$ and permutes the other three cyclically. Let $r'$ be its second fixed point, $z=(t-r)/(t-r')$ (so that $g$ acts by $z\mapsto\omega z$ or $z\mapsto\omega^2z$), and $X=z^3$. Write
\[
  F \;=\; \varphi_0(X)+\varphi_1(X)\,z+\varphi_2(X)\,z^2 .
\]
(If $r=\infty$ read $z=1/(t-r')$, and if $r'=\infty$ read $z=t-r$.) Then $\int\varphi_1(X)\,z\,\d t/\sqrt R$ is elementary: there are constants $\kappa,K,W$ with $R=\kappa\,(t-r')^4\,z\,(X-K)$, $\d z=W\,\d t/(t-r')^2$ (the factors $t-r'$ being replaced by $1$ when $r'=\infty$), and
\begin{equation}\label{eq:order3-reduction}
  \int \varphi_1(X)\,z\,\frac{\d t}{\sqrt R}\;=\;\frac{1}{3W}\int\frac{\varphi_1(X)\,\d X}{\sqrt{\kappa\,X(X-K)}} .
\end{equation}
The components $\varphi_0$ and $\varphi_2z^2$ descend to curves of genus $1$.
\end{theorem}

\begin{proof}
In the coordinate $z$ the branch points are $0$ and an orbit $\{\zeta,\omega\zeta,\omega^2\zeta\}$, which gives the form of $R$ with $K=\zeta^3$. Then $\sqrt R=(t-r')^2\sqrt{\kappa z(X-K)}$ and $\d X=3z^2\,\d z$, so $z\,\d t/\sqrt R=\d X/\bigl(3W\,z\sqrt{\kappa z(X-K)}\bigr)$ and $z\sqrt z=\sqrt X$. For the other two components the same computation gives the radicals $X^{5/6}(X-K)^{1/2}$ and $X^{1/6}(X-K)^{1/2}$, which define curves of genus $1$ by Riemann--Hurwitz.
\end{proof}

The second fixed point is rational in $r$: $r'$ is the other root of the Hessian of the cubic $R/(t-r)$ (cf.\ Lemma~\ref{lem:cyclic-mobius} and Proposition~\ref{prop:invariant} below). For a cubic $R$ with $b^2=3ac$ one may take $r=\infty$, $r'=-b/(3a)$, and the theorem is the observation that $R=a(t-r')^3+c_0$ is a polynomial in $(t-r')^3$.

In terms of Section~\ref{sec:geometric}: the lift of $g$ fixing the point over $r$ is an automorphism of order $3$ with fixed points, and $\varphi_1z\,\d t/y$ is its invariant part.

\begin{corollary}[Goursat's twelve-term criterion]\label{cor:A4}
In the situation of Theorem~\ref{thm:order3}, decompose $F=F^{(0)}+F^{(1)}+F^{(2)}+F^{(3)}$ as in Theorem~\ref{thm:goursat2} and $F^{(0)}=\varphi_0(X)+\varphi_1(X)z+\varphi_2(X)z^2$. If $\varphi_0=\varphi_2=0$ then $\int F\,\d t/\sqrt R$ is elementary.
\end{corollary}

Because $A_4=V_4\rtimes\Z/3$, all elements of order $3$ act in the same way on $V_4$-invariant functions (up to inversion), so a single $g$ suffices; this is the content of Goursat's conditions (6). His examples
\[
  \int\frac{(t^3-4)\,\d t}{(t^3+8)\sqrt{t^3-1}},\qquad \int\frac{t^2\,\d t}{(t^3+8)\sqrt{t^3-1}},\qquad \int\frac{t\,\d t}{(t^3+8)\sqrt{t^3-1}}
\]
(the last being the integral of Legendre and Clausen) are obtained by him from Theorem~\ref{thm:order3} applied at the three finite roots $r=1,\omega,\omega^2$ of $t^3-1$, where $z=(t-r)/(t+2r)$. Alternatively, the first and third satisfy~\eqref{eq:V4-sum} and follow from Theorem~\ref{thm:goursat2} alone, while the second is $t^2\varphi(t^3)\,\d t/\sqrt{t^3-1}$, the case $r=\infty$ of Theorem~\ref{thm:order3}.

\begin{proposition}[Order four gives nothing new]\label{prop:order4}
Suppose the branch points are harmonic, so that their M\"obius group is dihedral of order $8$, containing $V_4$ with index $2$. Let $h$ be an element outside $V_4$ (of order $4$, or an involution fixing two branch points). Then no lift of $h$ to $y^2=R$ produces elementary integrals beyond those of Theorem~\ref{thm:goursat2}.
\end{proposition}

\begin{proof}
Take a branch point as origin of the group law on $C$, so that the hyperelliptic involution is $-1$. An automorphism of $C$ descends to $\mathbb{P}^1$ iff it commutes with $-1$, i.e.\ iff it has the form $P\mapsto uP+p$ with $u$ a unit and $2p=0$; it induces an element of $V_4$ exactly when $u=\pm1$. So any lift $\widetilde h$ of $h\notin V_4$ has $u=\pm i$, and $\widetilde h^*(\d t/y)=\pm i\,\d t/y$. Thus $F\,\d t/y$ is $\widetilde h$-invariant iff $F\circ h=\mp i\,F$. If $h$ is an involution this forces $F=0$. If $h$ has order $4$ it gives $F\circ h^2=-F$ with $h^2\in V_4$ a pairing involution, so $F$ is already covered by Theorem~\ref{thm:goursat1}.
\end{proof}

This is Goursat's remark that ``the substitutions of period $4$ do not provide new pseudo-elliptic integrals''. Concretely, for $R=t^4-1$ the $V_4$-invariant functions are spanned over $\C(t^4)$-type invariants by even powers of $t$, and $\int t^{2j}\varphi(t^4)\,\d t/\sqrt{t^4-1}$ leads to the radicals $x^{3/4}(x-1)^{1/2}$ and $x^{1/4}(x-1)^{1/2}$, both of genus $1$.

\section{The cube-root analog}
\label{sec:cuberoot}

We now consider integrands of the form
\begin{equation}\label{eq:cuberoot-integrand}
  J = \int F(t)\,\frac{\d t}{\sqrt[3]{R(t)}}, \qquad F\in\C(t),\ R\in\C[t].
\end{equation}
The natural geometric setting is the cyclic cubic cover $X_R \colon y^3 = R(t)$.

\subsection{Genus arithmetic}

A direct application of Riemann--Hurwitz to the projection $X_R\to\mathbb{P}^1$, $(t,y)\mapsto t$, of degree $3$ yields the genus of $X_R$ for $R$ of degree $n$ with simple roots:
\[
  g(X_R) \;=\;
  \begin{cases}
    n-1 & \text{if } n\not\equiv 0\pmod 3, \\
    n-2 & \text{if } n\equiv 0\pmod 3.
  \end{cases}
\]
The case relevant to us is genus $1$:
\begin{center}
\begin{tabular}{c|ccccc}
  $\deg R$ & 1 & 2 & 3 & 4 & 5 \\\hline
  $g(X_R)$ & 0 & 1 & 1 & 3 & 4
\end{tabular}
\end{center}
We will focus on the genus-$1$ cases $\deg R\in\{2,3\}$. The case $\deg R=3$ is the cleanest analog of Goursat's quartic-square-root setting (also genus $1$), and we treat it first; the case $\deg R = 2$ may be reduced to $\deg R = 3$ by viewing one root as lying at infinity.

\subsection{Order-3 M\"obius substitutions}

Suppose $R$ is cubic with simple roots $\{a,b,c\}\subset\C$.

\begin{lemma}\label{lem:cyclic-mobius}
There exists a unique M\"obius transformation $S\in\PGL_2(\C)$ of order exactly $3$ sending $a\mapsto b\mapsto c\mapsto a$. Its inverse $S^{-1}=S^2$ sends $a\mapsto c\mapsto b\mapsto a$. The two fixed points $\alpha,\beta\in\mathbb{P}^1(\C)$ of $S$ are not among $\{a,b,c\}$.
\end{lemma}

\begin{proof}
A M\"obius transformation is determined by its action on three points, so $S$ exists and is unique. To see $S$ has order exactly $3$: it is non-trivial since $S(a)=b\neq a$, and $S^3(a)=S^2(b)=S(c)=a$ together with the analogous identities on $b,c$ show that $S^3$ fixes three distinct points $\{a,b,c\}$, hence $S^3=\mathrm{id}$. So $\operatorname{ord}(S)$ divides $3$, and is therefore exactly $3$.\\

For the fixed points: if $S$ fixed some $r\in\{a,b,c\}$, then $S(r)=r$ would contradict the cyclic permutation $a\mapsto b\mapsto c\mapsto a$ (which has no fixed point in $\{a,b,c\}$).
\end{proof}

In coordinates adapted to $S$, the cyclic action takes a particularly simple form.

\begin{lemma}\label{lem:cyclic-canonical}
Let $\alpha,\beta$ be the fixed points of $S$ in Lemma~\ref{lem:cyclic-mobius}, and set
\[
  z := \frac{t-\alpha}{t-\beta}.
\]
Then $S$ acts in $z$-coordinates as $z\mapsto \omega z$, where $\omega=e^{2\pi i/3}$, and there exist non-zero constants $c,K\in\C$ with
\begin{equation}\label{eq:R-as-z3-K}
  R(t(z))\,(1-z)^3 \;=\; c\,(z^3 - K).
\end{equation}
\end{lemma}

\begin{proof}
The substitution $z=(t-\alpha)/(t-\beta)$ sends $\alpha\mapsto 0$ and $\beta\mapsto\infty$, conjugating $S$ to a M\"obius transformation of $z$ fixing $\{0,\infty\}$, hence multiplication by a scalar $\lambda$. The order condition $\lambda^3=1$ together with $\lambda\neq 1$ forces $\lambda\in\{\omega,\omega^2\}$; relabelling $S\leftrightarrow S^2$ if necessary, we take $\lambda=\omega$.\\

The inverse change of variable is $t(z)=(\alpha-\beta z)/(1-z)$. As in Step~2 of the proof of Theorem~\ref{thm:goursat1}, multiplying through by $(1-z)^3$ converts $R(t(z))$ into a polynomial of degree $3$ in $z$. Its roots in $z$-coordinates are the images of $a,b,c$ under $r\mapsto(r-\alpha)/(r-\beta)$. Because $S$ acts as multiplication by $\omega$ and permutes $\{a,b,c\}$ cyclically, those images form a single $\omega$-orbit $\{\zeta,\omega\zeta,\omega^2\zeta\}$. Hence
\[
  R(t(z))(1-z)^3 \;=\; c\,(z-\zeta)(z-\omega\zeta)(z-\omega^2\zeta) \;=\; c\,(z^3 - \zeta^3),
\]
and $K=\zeta^3$.
\end{proof}

\subsection{The eigendecomposition}

In $z$-coordinates the cube-radical differential takes a particularly simple form. Writing $t=t(z)=(\alpha-\beta z)/(1-z)$, $\d t=(\alpha-\beta)\,\d z/(1-z)^2$, and using~\eqref{eq:R-as-z3-K}:
\begin{equation}\label{eq:H-defn}
  \frac{F(t)\,\d t}{R(t)^{1/3}} \;=\; \frac{(\alpha-\beta)\,F(t(z))\,\d z\,/(1-z)^2}{c^{1/3}(z^3-K)^{1/3}/(1-z)} \;=\; H(z)\,\frac{\d z}{(z^3-K)^{1/3}},
\end{equation}
where
\begin{equation}\label{eq:H-formula}
  H(z) \;:=\; \frac{(\alpha-\beta)\,F(t(z))}{c^{1/3}\,(1-z)}.
\end{equation}
Here we have chosen a specific cube root $c^{1/3}$ once and for all; different choices alter $H$ and the radical $(z^3-K)^{1/3}$ by mutually cancelling third roots of unity.\\

The substitution $\sigma\colon z\mapsto \omega z$ generates a cyclic group of order $3$ acting on $\C(z)$. Under this action, $\C(z)$ decomposes as a direct sum of three eigenspaces:
\[
  \C(z) \;=\; V_0 \oplus V_1 \oplus V_2,\qquad V_k = \{f \in \C(z) : f(\omega z) = \omega^k f(z)\}.
\]
The eigenspaces are characterized by
\[
  V_0 = \C(z^3),\quad V_1 = z\cdot\C(z^3),\quad V_2 = z^2\cdot\C(z^3),
\]
and the projection onto $V_k$ is
\begin{equation}\label{eq:projection-Vk}
  P_k f \;=\; \frac{1}{3}\bigl(f(z) + \omega^{-k}f(\omega z) + \omega^{-2k}f(\omega^2 z)\bigr).
\end{equation}

We can therefore write $H = H_0 + H_1 + H_2$ uniquely with $H_k = P_k H \in V_k$, and each piece contributes a separate integral
\[
  J_k := \int H_k(z)\,\frac{\d z}{(z^3-K)^{1/3}}.
\]
The integrals $J_0,J_1,J_2$ behave very differently from one another, and this is the main phenomenon.

\subsection{The main theorem}

\begin{theorem}[Cube-root reduction]\label{thm:cubic-goursat}
Let $R\in\C[t]$ be a cubic polynomial with simple roots, let $S\in\PGL_2(\C)$ be the order-$3$ M\"obius transformation cyclically permuting the roots of $R$, and let $H\in\C(z)$ be the function defined by~\eqref{eq:H-formula} for $F\in\C(t)$. Decompose $H=H_0+H_1+H_2$ as above, and write $H_k(z)=z^k\varphi_k(z^3)$ for unique rational functions $\varphi_k\in\C(x)$.
\begin{enumerate}[(i)]
  \item \textbf{(Two elementary pieces.)} $J_0$ and $J_2$ are elementary, with explicit antiderivatives obtained from rational parametrisations of the genus-zero curves $Y_0\colon y^3=x^2(x-K)$ and $Y_2\colon y^3=x-K$. In particular, if $H_1\equiv 0$ then $\int F\,\d t/\sqrt[3]{R}$ is elementary.
  \item \textbf{(The remaining piece carries the whole question.)} $\int F\,\d t/\sqrt[3]{R}$ is elementary if and only if the differential $\omega_1 = \varphi_1(x)\,\d x / \sqrt[3]{x(x-K)}$ on the genus-one curve $Y_1\colon y^3=x(x-K)$ has an elementary integral.
  \item \textbf{(Second-kind case.)} If $\omega_1$ is of the second kind, then $\int F\,\d t/\sqrt[3]{R}$ is elementary if and only if $\omega_1$ is exact, which is decided by Hermite reduction (Proposition~\ref{prop:hermite}).
\end{enumerate}
\end{theorem}

\begin{proof}

Substituting $x=z^3$ into each $J_k = \int H_k(z)\,\d z/(z^3-K)^{1/3}$ and using $\d x = 3z^2\,\d z$:
\begin{align*}
  J_0 &= \frac{1}{3}\int \frac{\varphi_0(x)\,\d x}{x^{2/3}(x-K)^{1/3}}, \qquad
  J_1 = \frac{1}{3}\int \frac{\varphi_1(x)\,\d x}{x^{1/3}(x-K)^{1/3}}, \qquad
  J_2 = \frac{1}{3}\int \frac{\varphi_2(x)\,\d x}{(x-K)^{1/3}}.
\end{align*}
The radicands $x^2(x-K)$, $x(x-K)$, $x-K$ define curves $Y_0, Y_1, Y_2$ respectively, and we treat each.

\smallskip\noindent
\textbf{$J_2$ is elementary.}
The radical $(x-K)^{1/3}$ defines a curve isomorphic to the affine line via $x=K+w^3$, after which $\d x = 3w^2\,\d w$ and
\[
  J_2 = \int w\,\varphi_2(K+w^3)\,\d w,
\]
which is the integral of a rational function in $w$ and hence elementary.

\smallskip\noindent
\textbf{$J_0$ is elementary.}
The radicand $x^2(x-K)$ has a \emph{double} root at $x=0$ and a \emph{simple} root at $x=K$. The cube-radical curve $Y_0\colon y^3 = x^2(x-K)$ has genus zero, parametrized rationally by
\begin{equation}\label{eq:Y0-param}
  x \;=\; \frac{K}{1-w^3}, \qquad y \;=\; \frac{Kw}{1-w^3},
\end{equation}
giving $\d x/y = 3w\,\d w/(1-w^3)$. Therefore
\[
  J_0 = \int \frac{\varphi_0(K/(1-w^3))\,w}{1-w^3}\,\d w,
\]
the integral of a rational function in $w$, hence elementary.

\smallskip\noindent
This proves (i).

\smallskip\noindent
\textbf{Proof of (ii).}
Let $W\colon w^3=z^3-K$, a curve of genus $1$ on which our integrand is $\Omega=H(z)\,\d z/w$, and let $\sigma\colon(z,w)\mapsto(\omega z,w)$. Suppose $\int\Omega$ is elementary, so that $\Omega=\d v+\sum_ic_i\,\d u_i/u_i$ with $v,u_i\in\C(W)$ by Theorem~\ref{thm:liouville-alg}. Then $(\sigma^j)^*\Omega$ has an expression of the same shape for $j=0,1,2$, and therefore so does each eigencomponent
\[
  \Omega_k \;=\; \tfrac13\sum_{j=0}^{2}\omega^{-jk}\,(\sigma^j)^*\Omega \;=\; H_k(z)\,\frac{\d z}{w}
\]
(complex constants in front of the logarithms are allowed). Conversely $\Omega=\Omega_0+\Omega_1+\Omega_2$. So $\int\Omega$ is elementary if and only if each $\int\Omega_k$ is, and by (i) this holds if and only if $\int\Omega_1$ is elementary. Finally $\Omega_1=\tfrac13\omega_1$ lies in the subfield $\C(Y_1)=\C(x,zw)\subset\C(W)$, and by Liouville's theorem an element of a differential field which has an elementary integral over a finite algebraic extension has one over the field itself.

\smallskip\noindent
\textbf{Proof of (iii).} This is Lemma~\ref{lem:second-kind} applied to $\omega_1$ on $Y_1$. Proposition~\ref{prop:hermite} applies to $Y_1$ with $n=3$, $k=1$ and radicand $x(x-K)$ of degree $2$.
\end{proof}

\begin{remark}[What (iii) says for polynomial $\varphi_1$]\label{rmk:f-poly-sufficient}
On $Y_1$ the branch points are $x=0,K,\infty$, each totally ramified, and $\d x/y$ is regular at the points over $0$ and $K$ with a double pole over $\infty$. Hence for a polynomial $\varphi_1$ the differential $\omega_1$ has a single pole and is of the second kind. But it need not be non-exact. From~\eqref{eq:exact-forms}, $\d(x^my^2)=\bigl[(m+\tfrac43)x^{m+1}-K(m+\tfrac23)x^m\bigr]\d x/y$, so modulo exact differentials
\[
  x^{m+1}\,\frac{\d x}{y}\;\equiv\;K\,\frac{3m+2}{3m+4}\;x^m\,\frac{\d x}{y},
\]
and for $\varphi_1=\sum_da_dx^d$:
\begin{equation}\label{eq:poly-criterion}
  \omega_1 \text{ is exact}\iff\sum_{d\ge0}a_d\,K^d\prod_{m=0}^{d-1}\frac{3m+2}{3m+4}=0 .
\end{equation}
For instance $\varphi_1=2x-K$ gives $\omega_1=\tfrac32\,\d(y^2)$, although $\varphi_1\neq0$ (Example~\ref{ex:exact-V1}).

The hypothesis that $\varphi_1$ be a polynomial is restrictive. It holds when $F$ is a polynomial \emph{and} $\beta=\infty$ (i.e.\ $R=c\,((t-\alpha)^3-K)$). For finite $\beta$ a polynomial $F$ has a pole at $z=1$ (the image of $t=\infty$), so $\varphi_1$ has a pole at $x=1$, which is not a branch point of $Y_1$ unless $K=1$; indeed $F\,\d t/\sqrt[3]{R}$ with $F$ polynomial generally has non-zero residues at the three points over $t=\infty$. Then $\omega_1$ is of the third kind and (iii) does not apply.
\end{remark}

\begin{remark}[Sufficient, not necessary]\label{rmk:cubic-not-necessary}
As in the square-root case (Remark~\ref{rmk:not-necessary}), $H_1\not\equiv0$ does \emph{not} imply that the integral is non-elementary. There are three distinct reasons, each illustrated in Section~\ref{sec:examples}: $\omega_1$ may be exact (Example~\ref{ex:exact-V1}); $\omega_1$ may be elementary because of a further symmetry, the transpositions of Section~\ref{sec:transpositions} (Examples~\ref{ex:F1-symmetric} and~\ref{ex:combined}); or $\omega_1$ may be elementary of the third kind for arithmetic reasons (a torsion condition), which no symmetry argument detects.
\end{remark}

\begin{remark}[Comparison with the square-root case]\label{rmk:same-dichotomy}
The structure is exactly that of Section~\ref{sec:geometric}. On $X_R\colon y^3=R$ let $\rho\colon y\mapsto\omega y$ and let $\sigma$ be a lift of $S$. The group $\langle\rho,\sigma\rangle\cong(\Z/3)^2$ contains exactly two non-trivial translations, $\tau^{\pm1}$, and six elements with fixed points. Each of the three eigencomponents of $F\,\d t/y$ is invariant under one of the three subgroups $\langle\sigma\rho^a\rangle$; for two of them the generator has fixed points and the quotient is rational ($Y_0$, $Y_2$); for the third it is the translation $\tau$ and the quotient $Y_1$ has genus $1$. For a single involution in the square-root case one finds the same thing with two components instead of three: the anti-invariant part descends to a conic, the invariant part to a curve of genus $1$.

Moreover the obstruction curve is again the original curve. Every curve $y^3=R$ of genus $1$ has the automorphism $\rho$ of order $3$ with fixed points, hence $j$-invariant $0$; so does $Y_1$; hence $Y_1\cong X_R$ over $\C$, whatever the value of $K$, and $X_R\to Y_1$ is the endomorphism $1-\omega$ up to isomorphism.
\end{remark}

\begin{remark}[Degree-$2$ radicand]\label{rmk:degree-2}
For $R$ of degree $2$, the cube-radical curve $y^3 = R(t)$ still has genus $1$, and Theorem~\ref{thm:cubic-goursat} applies without modification. The proof of Lemma~\ref{lem:cyclic-canonical} extends uniformly: when $\deg R = 2$, only two factors $t(z) - r_i$ contribute and produce a $1/(1-z)^2$ rather than $1/(1-z)^3$, but multiplying through by $(1-z)^3$ still gives a polynomial of degree $3$ in $z$. The "extra" factor $(1-z)$ encodes the implicit ramification point at $\infty$ of the projection $X_R\to\mathbb{P}^1$.

Two structural features specific to $\deg R = 2$ deserve emphasis.

\medskip\noindent
\textbf{(a) $K = 1$ is forced.} Because $S$ cyclically permutes the two finite roots and the point at $\infty$, the three roots in $z$-coordinates form a $\langle\omega\rangle$-orbit on $\mathbb{P}^1(\C)$ containing $z(\infty) = 1$, hence are exactly $\{1,\omega,\omega^2\}$. Therefore
\[
  R(t(z))(1-z)^3 \;=\; c\,(z^3 - 1).
\]
For cubic $R$, by contrast, $K$ depends on the configuration of the roots and may be any non-zero complex number.

\medskip\noindent
\textbf{(b) The fixed points of $S$ are finite.} A direct computation shows that the unique cyclic M\"obius transformation permuting $\{r_1, r_2, \infty\}$ as $r_1 \to r_2 \to \infty \to r_1$ is
\[
  S(t) \;=\; \frac{r_1\,t - (r_1^2 - r_1 r_2 + r_2^2)}{t - r_2},
\]
with fixed-point equation $t^2 - (r_1+r_2)\,t + (r_1^2 - r_1 r_2 + r_2^2) = 0$ and discriminant $-3(r_1-r_2)^2\neq0$. Hence
\[
  \alpha,\beta \;=\; \frac{r_1+r_2}{2} \;\pm\; \frac{i\sqrt{3}}{2}\,(r_1-r_2)
\]
are distinct and finite (they are non-real when $r_1,r_2$ are real, and real when $r_1,r_2$ are complex conjugates). In particular, the change of variable $z = (t-\alpha)/(t-\beta)$ is unavoidably non-trivial --- there is no degree-$2$ radicand for which it degenerates to $z = t - \alpha$, in contrast to the cubic case (e.g.\ $R = t^3 - 1$).
\end{remark}

The two features above combine to force a uniform conclusion when $F = 1$:

\begin{corollary}[Non-elementarity for quadratic radicand]\label{cor:degree2-nonelementary}
Let $R(t)\in\C[t]$ be a quadratic polynomial with two distinct roots. Then
\[
  \int \frac{\d t}{\sqrt[3]{R(t)}}
\]
is \emph{not} elementary.
\end{corollary}

\begin{proof}
After an affine change of variable $R=c\,(t^2-1)$. On $y^3=c(t^2-1)$ take $y$ as coordinate: $2ct\,\d t=3y^2\,\d y$, so
\[
  \frac{\d t}{y}\;=\;\frac{3}{2c}\,\frac{y\,\d y}{t},\qquad t^2=\frac{y^3}{c}+1 .
\]
This is the classical differential of the second kind on the elliptic curve $t^2=y^3/c+1$: it is regular at finite points and has a double pole at the point at infinity. It is not exact, since a function with a single simple pole does not exist on a curve of genus $1$. By Lemma~\ref{lem:second-kind} its integral is not elementary.

Alternatively: $F=1$ is already Hermite-reduced ($\deg P\le d-2=0$), the point at infinity is a branch point, and Proposition~\ref{prop:hermite} applies directly. Or one may invoke Chebyshev's theorem with $m=0$, $n=2$, $p=-\tfrac13$.
\end{proof}

\begin{remark}
The contrast with the square-root case is striking: every integral $\int\d t/\sqrt{R(t)}$ with $\deg R = 2$ is elementary (yielding $\arcsin$, $\operatorname{arccosh}$, or a logarithm depending on the sign pattern), because the curve $y^2 = R(t)$ has genus $0$. The cube-radical curve $y^3 = R(t)$ has genus $1$, and Corollary~\ref{cor:degree2-nonelementary} shows that for $F=1$ the obstruction in Theorem~\ref{thm:cubic-goursat}(ii) is generic — indeed, unavoidable. The point is that under the (unavoidable) substitution $z = (t-\alpha)/(t-\beta)$, the constant function $F(t) = 1$ does not remain constant: it acquires a pole at $z = 1$ (corresponding to $t = \infty$), and that single pole is enough to populate all three eigencomponents of $H$.
\end{remark}

The same eigenspace machinery applies to the dual exponent $2/3$, and the two exponents together cover the full function field $\C(t,y)$ with $y^3=R(t)$.

\begin{corollary}[Cube-root analog at exponent $2/3$ and full algebraic field]\label{cor:two-thirds}
Let $R\in\C[t]$ be a cubic polynomial with simple roots, $S$ the order-$3$ M\"obius transformation of Lemma~\ref{lem:cyclic-mobius}, $\alpha,\beta$ its fixed points, $z=(t-\alpha)/(t-\beta)$, and $c, K$ the constants of Lemma~\ref{lem:cyclic-canonical}. For $F\in\C(t)$, define
\begin{equation}\label{eq:Htilde-formula}
  \widetilde H(z) \;:=\; \frac{(\alpha-\beta)\,F(t(z))}{c^{2/3}}
\end{equation}
(with the convention that the prefactor $\alpha-\beta$ is replaced by $1$ when $\beta=\infty$), and decompose $\widetilde H = \widetilde H_0 + \widetilde H_1 + \widetilde H_2$ into $\omega$-eigencomponents under $z\mapsto\omega z$, with $\widetilde H_k(z) = z^k\widetilde\varphi_k(z^3)$.
\begin{enumerate}[(i)]
  \item \emph{(Sufficient criterion for elementarity at $2/3$.)} If $\widetilde H_0\equiv 0$, then the integral $\int F(t)\,\d t/R(t)^{2/3}$ is elementary, with the explicit reduction
  \begin{equation}\label{eq:two-thirds-reduction}
    \int \frac{F(t)\,\d t}{R(t)^{2/3}}
    \;=\; -\!\int \frac{\widetilde\varphi_1\!\bigl(\tfrac{Ks^3}{s^3-1}\bigr)\,s\,\d s}{s^3-1}
    \;+\; \int \widetilde\varphi_2(u^3+K)\,\d u,
  \end{equation}
  both pieces being rational integrals in the parameter variables $s$ and $u$.

  \item \emph{(The remaining piece at $2/3$.)} $\int F\,\d t/R^{2/3}$ is elementary if and only if the differential $\widetilde\omega_0 = \widetilde\varphi_0(x)\,\d x/[x^{2/3}(x-K)^{2/3}]$ on the curve $\widetilde Y_0\colon y^3 = x^2(x-K)^2$ (birational to $Y_1\colon v^3 = x(x-K)$) has an elementary integral. If $\widetilde\omega_0$ is of the second kind this holds if and only if it is exact.
  \item \emph{(Full algebraic field.)} Every $G\in\C(t,y)$ with $y^3=R(t)$ admits a unique representation $G = G_0(t) + G_1(t)\,y + G_2(t)\,y^2$ with $G_i\in\C(t)$, and
  \begin{equation}\label{eq:cubic-field-split}
    \int G\,\d t \;=\; \int G_0\,\d t \;+\; \int \frac{G_2\,R\,\d t}{R^{1/3}} \;+\; \int \frac{G_1\,R\,\d t}{R^{2/3}}.
  \end{equation}
  The first piece is elementary, and $\int G\,\d t$ is elementary if and only if the second and the third are; these are of the two types treated above.
\end{enumerate}
\end{corollary}

\begin{proof}
\textbf{(i)} Since $R(t)(1-z)^3 = c(z^3-K)$ by Lemma~\ref{lem:cyclic-canonical} and $\d t = (\alpha-\beta)\,\d z/(1-z)^2$, one computes
\[
  R(t)^{2/3} \;=\; \frac{c^{2/3}\,(z^3-K)^{2/3}}{(1-z)^2}
  \quad\Longrightarrow\quad
  \frac{\d t}{R(t)^{2/3}} \;=\; \frac{(\alpha-\beta)\,\d z}{c^{2/3}\,(z^3-K)^{2/3}},
\]
so the $(1-z)^2$ factors cancel exactly --- in contrast to~\eqref{eq:H-formula} for the $1/3$-exponent, no $1/(1-z)$ factor remains. With $\widetilde H$ as in~\eqref{eq:Htilde-formula},
\[
  \frac{F(t)\,\d t}{R(t)^{2/3}} \;=\; \frac{\widetilde H(z)\,\d z}{(z^3-K)^{2/3}}.
\]
Substituting $x=z^3$ in each eigencomponent gives
\[
  \int \frac{\widetilde H_k(z)\,\d z}{(z^3-K)^{2/3}}
  \;=\; \widetilde J_k \;:=\; \int \frac{\widetilde\varphi_k(x)\,\d x}{3\,x^{(2-k)/3}\,(x-K)^{2/3}}, \qquad k\in\{0,1,2\},
\]
so the integral lives on the curve $\widetilde Y_k\colon y^3 = x^{2-k}(x-K)^2$. We treat each $k$:
\begin{itemize}

  \item $\widetilde Y_0\colon y^3 = x^2(x-K)^2 = [x(x-K)]^2$. With $v=x(x-K)/y$ one has $v^3 = x(x-K)$, which identifies $\widetilde Y_0$ with $Y_1$, and $\widetilde J_0=\tfrac13\int\widetilde\varphi_0(x)\,v\,\d x/(x(x-K))$ is an Abelian integral on this curve of genus $1$.

  \item $\widetilde Y_1\colon y^3 = x(x-K)^2$. The rational parametrization $x = Ks^3/(s^3-1)$, $y=Ks/(s^3-1)$ has Jacobian $\d x = -3Ks^2\,\d s/(s^3-1)^2$ and gives $x^{1/3}(x-K)^{2/3} = Ks/(s^3-1)$ (with consistent cube-root branches), whence
  \[
    \widetilde J_1 \;=\; -\!\int \frac{\widetilde\varphi_1\!\bigl(\tfrac{Ks^3}{s^3-1}\bigr)\,s\,\d s}{s^3-1}.
  \]
  \item $\widetilde Y_2\colon y^3 = (x-K)^2$. The substitution $u = (x-K)^{1/3}$ gives $x=u^3+K$, $\d x = 3u^2\,\d u$, $(x-K)^{2/3} = u^2$, whence
  \[
    \widetilde J_2 \;=\; \int \widetilde\varphi_2(u^3+K)\,\d u.
  \]
\end{itemize}

The reduction~\eqref{eq:two-thirds-reduction} follows. That the integral is elementary if and only if $\widetilde J_0$ is follows by the averaging argument in the proof of Theorem~\ref{thm:cubic-goursat}(ii), and the last sentence of (ii) is Lemma~\ref{lem:second-kind}.

\smallskip
\noindent\textbf{(iii)} The decomposition $G=G_0+G_1y+G_2y^2$ is the unique representation of $\C(t,y)$ as a free $\C(t)$-module on the basis $\{1,y,y^2\}$, and
\[
  G_1\,y \;=\; \frac{G_1 R}{y^2} \;=\; \frac{G_1 R}{R^{2/3}},
  \qquad
  G_2\,y^2 \;=\; \frac{G_2 R}{y} \;=\; \frac{G_2 R}{R^{1/3}},
\]
which is~\eqref{eq:cubic-field-split}. The three pieces are the eigencomponents of $G\,\d t$ under $\rho\colon y\mapsto\omega y$, so by the same averaging argument (with $\rho$ in place of $\sigma$) their sum is elementary if and only if each of them is.
\end{proof}

\begin{remark}[A duality between the two exponents]
Comparing~\eqref{eq:H-formula} and~\eqref{eq:Htilde-formula}, the obstructive eigenspace shifts from $V_1$ at exponent $1/3$ to $V_0$ at exponent $2/3$. Geometrically, this reflects which $1$-form of $\C(t,y)$ is of the first kind: on the cube-root elliptic curve $y^3=x(x-K)$, the differential $\d x/y^2$ is the unique (up to scale) holomorphic form, while $\d x/y$ has a pole at infinity. The eigenspace assignment of the obstruction tracks precisely this dichotomy.
\end{remark}

\subsection{The transpositions, and a formulation through covariants}\label{sec:transpositions}

The symmetry group of three points on $\mathbb{P}^1$ is $S_3$, not $\Z/3$. Besides $S$ and $S^{-1}$ it contains three involutions $T$, each fixing one root of $R$ and exchanging the other two. Such a $T$ lifts to an involution $\widetilde T$ of $X_R$ which fixes the point over the fixed root; so $\widetilde T$ has fixed points, $X_R/\widetilde T\cong\mathbb{P}^1$, and by the principle of Section~\ref{sec:geometric} every $\widetilde T$-invariant differential integrates elementarily, \emph{whatever its $\Z/3$-eigencomponents}. The group generated by $\rho$, $\sigma$ and $\widetilde T$ has order $18$, and the genuinely obstructive space is only half of $V_1$ (resp.\ of $V_0$ at exponent $2/3$).

This is best expressed through the classical covariants of the cubic, which at the same time removes the roots of $R$, the cube roots of unity, the constants $c^{1/3}$ and $K$, and the fixed points from the statement. For $R=at^3+bt^2+ct+d$ (we allow $a=0$, so as to include quadratic radicands) put
\begin{equation}\label{eq:covariants}
  \mathcal H \;=\;(b^2-3ac)\,t^2+(bc-9ad)\,t+(c^2-3bd),\qquad \mathcal G\;=\;3\mathcal H'R-2\mathcal HR',
\end{equation}
the Hessian and the Jacobian covariant. They satisfy the syzygy
\begin{equation}\label{eq:syzygy}
  \mathcal G^2\;=\;4\mathcal H^3-27\,\Delta\,R^2,\qquad \Delta=\operatorname{disc}(R),
\end{equation}
and $\operatorname{disc}(\mathcal H)=-3\Delta$. Let
\[
  x_0=\frac{\mathcal G}{R},\qquad y=R^{1/3},\qquad w=\frac{\mathcal H}{y^2}.
\]

\begin{proposition}\label{prop:invariant}
\begin{enumerate}[(a)]
\item The fixed points $\alpha,\beta$ of $S$ are the roots of $\mathcal H$ (one of them is $\infty$ when $\deg\mathcal H=1$); they lie in $\Q(\sqrt{-3\Delta})$ when $R\in\Q[t]$. The function $x_0$ has degree $3$, is invariant under $S$, and generates the field of $S$-invariant functions; each involution $T$ acts on it by $x_0\mapsto-x_0$.
\item \emph{(Exponent $1/3$.)} Every $F\in\C(t)$ has a unique expression
\begin{equation}\label{eq:decomp13}
  F\;=\;\frac{A(x_0)+B(x_0)\,t}{\mathcal H}\;+\;C(x_0)\,\frac{\mathcal H}{R},
\end{equation}
found by solving a $3\times3$ linear system over $\C(x_0)$. The first term is the $V_0\oplus V_2$-part of Theorem~\ref{thm:cubic-goursat} and the second its $V_1$-part. Thus $J_1=0$ iff $C=0$ iff $F\mathcal H\in\C(x_0)+\C(x_0)\,t$. The component of the first term along $1/(t-\alpha)$ becomes a rational integral in $u=y/(t-\beta)$, and that along $1/(t-\beta)$ a rational integral in $u=y/(t-\alpha)$.
\item \emph{(Exponent $2/3$.)} Every $F$ has a unique expression $F=C(x_0)+\bigl(A(x_0)\,q_1+B(x_0)\,q_2\bigr)/\mathcal H$, where $q_1,q_2$ span the quadratics apolar to $\mathcal H$ (for $\mathcal H\propto t^2-st+p$: $q_1=t^2-p$, $q_2=2t-s$). Here $C=\tfrac13\operatorname{Tr}_{\C(t)/\C(x_0)}F$ is the $\widetilde V_0$-part of Corollary~\ref{cor:two-thirds}, and the component along $(t-\gamma)^2/\mathcal H$ becomes a rational integral in $u=y/(t-\gamma)$, $\gamma\in\{\alpha,\beta\}$.
\item \emph{(Transpositions.)} Write $C(x_0)=x_0\,\psi(x_0^2)+\chi(x_0^2)$. The odd part is elementary at both exponents:
\begin{equation}\label{eq:odd-part}
  \int\psi(x_0^2)\,x_0\,\frac{\mathcal H}{R}\,\frac{\d t}{y}\;=\;3\!\int\psi(4w^3-27\Delta)\,w\,\d w,
  \qquad
  \int\psi(x_0^2)\,x_0\,\frac{\d t}{y^2}\;=\;3\!\int\psi(4w^3-27\Delta)\,\d w .
\end{equation}
Only the even part $\chi(x_0^2)$ of $C$ remains.
\end{enumerate}
\end{proposition}

\begin{proof}
(a) In the coordinate $z$ of Lemma~\ref{lem:cyclic-canonical} the cubic becomes $z^3-K$ up to the factor $(1-z)^3$, with Hessian proportional to $z$, i.e.\ vanishing at $z=0,\infty$; covariance gives the first claim, and the field of definition follows from $\operatorname{disc}\mathcal H=-3\Delta$. Similarly $\mathcal G$ becomes proportional to $z^3+K$, so $x_0$ is a M\"obius function of $z^3$. By~\eqref{eq:syzygy}, $x_0^2=4\mathcal H^3/R^2-27\Delta$, and $\mathcal H^3/R^2$ is invariant under every M\"obius map preserving the roots; so $x_0\circ T=\pm x_0$, and the sign is $-$ because a function of degree $3$ cannot be invariant under a group of order $6$.

(b) By~\eqref{eq:H-formula}, $F$ has $H\in V_k$ iff $F\in\C(x)\,z^k(1-z)$, and $1-z=(\alpha-\beta)/(t-\beta)$. So $V_0\leftrightarrow\C(x_0)/(t-\beta)$, $V_2\leftrightarrow\C(x_0)\,z^2/(t-\beta)=\C(x_0)/(t-\alpha)$ (because $z^3\in\C(x_0)$), and $V_1\leftrightarrow\C(x_0)\,(t-\alpha)/(t-\beta)^2=\C(x_0)\,\mathcal H/R$ (because $R/(t-\beta)^3\in\C(x_0)$). The sum of the first two is $\{(A+Bt)/\mathcal H\}$. The parameters are those of the proof of Theorem~\ref{thm:cubic-goursat}: $w\propto y/(t-\alpha)$ on $Y_0$ and $w\propto y/(t-\beta)$ on $Y_2$.

(c) Likewise $\widetilde V_k\leftrightarrow\C(x_0)\,z^k$, and $z=(t-\alpha)^2/\bigl((t-\alpha)(t-\beta)\bigr)$, $z^2\in\C(x_0)\,(t-\beta)^2/\bigl((t-\alpha)(t-\beta)\bigr)$; the squares $(t-\alpha)^2$, $(t-\beta)^2$ span the quadratics apolar to $\mathcal H$.

(d) Since $y'=R'y/(3R)$, $\d w=\bigl(\mathcal H'-\tfrac23\mathcal HR'/R\bigr)\d t/y^2=\tfrac13\,x_0\,\d t/y^2$, and $(\mathcal H/R)\,\d t/y=w\,\d t/y^2$. By~\eqref{eq:syzygy}, $x_0^2=4w^3-27\Delta$. Substituting gives~\eqref{eq:odd-part}.
\end{proof}

Part (d) is a direct computation, but its origin is geometric: $w$ and $x_0^2$ generate the function field of $X_R/\langle\tau,\widetilde T\rangle\cong\mathbb{P}^1$. For $R=t^3-1$ we have $\mathcal H=9t$, $\mathcal G=-27(t^3+1)$, and (d) splits $F=t\in V_1$ as
\[
  t\;=\;\frac{t^3+1}{2t^2}\;+\;\frac{t^3-1}{2t^2},\qquad \int\frac{t^3+1}{2t^2}\,\frac{\d t}{\sqrt[3]{t^3-1}}=\frac{(t^3-1)^{2/3}}{2t},
\]
the second summand being the genuine residual. The proposition also yields a family of elementary integrals attached to \emph{every} cubic: by (b),
\[
  \int\frac{(at+b)\,\d t}{\mathcal H(t)\,\sqrt[3]{R(t)}}\quad\text{is elementary for all constants } a,b .
\]

\subsection{A second level: the residual is a Goursat integral}\label{sec:second-level}

The M\"obius symmetries of three points are exhausted by $S_3$, and every cubic has them; so, unlike in Section~\ref{sec:order3}, there are no special radicands with further M\"obius substitutions, and no substitutions of order $4$ at all (the curve has $j=0$). The curve $X_R$ does have further automorphisms, but they mix $t$ and $y$. They are reached most simply through the residual.

\begin{proposition}\label{prop:second-level}
In the notation of Proposition~\ref{prop:invariant}, the residual differentials are
\[
  \chi(x_0^2)\,\frac{\mathcal H}{R}\,\frac{\d t}{y}\;=\;\frac{3\,w\,\chi(Y^2)\,\d w}{Y},
  \qquad
  \chi(x_0^2)\,\frac{\d t}{y^2}\;=\;\frac{3\,\chi(Y^2)\,\d w}{Y},
  \qquad Y^2=4w^3-27\Delta ,
\]
where $w=\mathcal H/y^2$ and $Y=x_0=\mathcal G/R$. Thus the residual of the cube-root problem is a square-root integral $\int\Phi(w)\,\d w/\sqrt{4w^3-27\Delta}$ on a cubic with $j=0$, in which the square root is a rational function of $t$.
\end{proposition}

\begin{proof}
This is the computation in the proof of Proposition~\ref{prop:invariant}(d), where it was shown that $\d t/y^2=3\,\d w/x_0$, that $(\mathcal H/R)\,\d t/y=w\,\d t/y^2$, and that $x_0^2=4w^3-27\Delta$.
\end{proof}

All of Section~\ref{sec:goursat} therefore applies to the residual: Hermite reduction, the three involutions of Theorem~\ref{thm:goursat2}, which are indexed by the roots of $4w^3=27\Delta$, and the order-$3$ substitutions of Theorem~\ref{thm:order3} at those roots. Pulled back to $X_R$ these are substitutions of order $2$ and $3$ that are rational in $t$ and $y$ jointly but not M\"obius in $t$. (The order-$3$ substitution at $w=\infty$ gives nothing, since it is the translation $\tau$ already divided out.) Geometrically, $(w,Y)$ are coordinates on $Y_1\cong X_R/\langle\tau\rangle$, and the new symmetries are the automorphisms of $X_R$ of the form $P\mapsto uP+p$ with $p$ a point of order $3$ outside the three branch points. The process can be iterated --- what the second level leaves is again a problem on an isomorphic curve --- and it remains a sufficient criterion at every level. Example~\ref{ex:second-level} shows that it does detect new elementary integrals.

\subsection{Relation to Goursat's order-three substitutions and to Chebyshev's theorem}\label{sec:relations}

\emph{Goursat.} As noted in Remark~\ref{rmk:same-dichotomy}, every genus-$1$ curve $y^3=R$ has $j=0$ and is therefore birational to $Y^2=X^3+c$, the curve of Theorem~\ref{thm:order3}. For quadratic $R$ this is visible at once: $y^3=t^2-1$ is $t^2=y^3+1$ with the roles of the coordinates exchanged, and
\[
  \int F(t)\,\frac{\d t}{\sqrt[3]{t^2-1}}\;=\;\frac32\int F\bigl(\sqrt{y^3+1}\bigr)\,\frac{y\,\d y}{\sqrt{y^3+1}} .
\]
For even $F$ this is a Goursat integral on $Y^2=y^3+1$, and the order-$3$ map $S$ of this section corresponds to an order-$3$ automorphism of that curve. The two theories are thus two coordinate descriptions of the automorphisms of one curve. They are not literally the same theorem, because the classes of integrands differ: $F(t)\,\d t/y$ with $F$ rational in $t$ on one side, $F(X)\,\d X/Y$ with $F$ rational in $X$ on the other.

\emph{Chebyshev.} For $R=t^3-K$ we have $\beta=\infty$, $z=t$, and for $F=t^m$ the eigencomponent is $k\equiv m\pmod 3$. Theorem~\ref{thm:cubic-goursat}(i) gives elementarity for $m\equiv0,2$, in agreement with Chebyshev's conditions $(m+1)/3-\tfrac13\in\Z$ and $(m+1)/3\in\Z$. For $m\equiv1$ Chebyshev's theorem says the integral is \emph{not} elementary; in our framework this follows from~\eqref{eq:poly-criterion}, since a single monomial $a_dx^d$ with $a_d\ne0$ never satisfies it. Corollary~\ref{cor:two-thirds} matches Chebyshev in the same way. What the present theory adds to Chebyshev's is rational $\varphi_k$ and, through Lemma~\ref{lem:cyclic-canonical}, general $R$.

\section{Worked examples}
\label{sec:examples}

We now illustrate the framework with a sequence of explicit calculations. Throughout, $\omega = e^{2\pi i/3}$.

\subsection{Square-root: a basic Goursat example}

\begin{example}\label{ex:goursat-basic}
Take $R(t) = (t^2-1)(t^2-4)$ and $F(t)=t$. The roots are $\{\pm 1, \pm 2\}$, and the three M\"obius involutions of Lemma~\ref{lem:involutions-pairs} are computed directly:
\[
  S_1(t) = -t,\qquad S_2(t) = 2/t,\qquad S_3(t) = -2/t.
\]
Verification of the $V_4$-sum:
\[
  F + F\!\circ\! S_1 + F\!\circ\! S_2 + F\!\circ\! S_3 \;=\; t + (-t) + \tfrac{2}{t} + \bigl(-\tfrac{2}{t}\bigr) \;=\; 0.
\]
The character decomposition (with the labelling of Theorem~\ref{thm:goursat2}'s proof, $\chi_j(S_j)=+1$) gives
\[
  F^{(1)} = 0,\qquad F^{(2)} = \frac{t^2+2}{2t},\qquad F^{(3)} = \frac{t^2-2}{2t}.
\]
(One checks directly: $F^{(2)}(2/t) = F^{(2)}(t)$ and $F^{(3)}(-2/t) = F^{(3)}(t)$.) Apply Theorem~\ref{thm:goursat1} to $F^{(3)}$ with the involution $S_1(t)=-t$ (fixed points $\alpha=0$, $\beta=\infty$): since $\beta=\infty$ the substitution is just $u=t$, and~\eqref{eq:R-becomes-even} becomes the trivial $R(u) = (u^2-1)(u^2-4)$, already even. Apply $x = u^2$ to obtain
\[
  \int F^{(3)}\,\frac{\d t}{\sqrt{R}} \;=\; \int \frac{t^2-2}{2t}\,\frac{\d t}{\sqrt{R(t)}}
  \;=\; \tfrac{1}{2}\int \frac{(x-2)\,\d x}{2x\sqrt{(x-1)(x-4)}}.
\]
A similar computation handles $F^{(2)}$. Adding the two contributions and simplifying yields
\[
  \int \frac{t\,\d t}{\sqrt{(t^2-1)(t^2-4)}} \;=\; \tfrac{1}{2}\log\!\Bigl(2t^2 - 5 + 2\sqrt{(t^2-1)(t^2-4)}\Bigr) + C,
\]
which one verifies by direct differentiation. \hfill$\square$
\end{example}

\begin{example}[The criterion is not necessary; an order-three substitution]\label{ex:sqrt-extra}
(a) With the same $R$, take $F=3t^4-10t^2+4$. Here $F^{(0)}\neq0$ (Remark~\ref{rmk:not-necessary}), so Theorem~\ref{thm:goursat2} does not apply. But Hermite reduction gives $F=A'R+\tfrac12AR'$ with $A=t$ and $\widetilde F=0$: the integral is $t\sqrt R$.

(b) For $R=t^3-1$ and $F=1/t$, Theorem~\ref{thm:order3} applies with $r=\infty$, $r'=0$, $z=1/t$, $X=t^{-3}$: $F=z$, so $\varphi_1=1$, and~\eqref{eq:order3-reduction} gives
\[
  \int\frac{\d t}{t\sqrt{t^3-1}}\;=\;\frac23\arctan\sqrt{t^3-1}+C .
\]
(c) For $R=t^3+1$ and $F=t$: $F$ is Hermite-reduced, and $t\,\d t/y$ is of the second kind. So $\int t\,\d t/\sqrt{t^3+1}$ is not elementary by Proposition~\ref{prop:hermite}. Note that Theorem~\ref{thm:goursat2} alone only tells us that $F^{(0)}=(t^4-8t)/(4(t^3+1))\neq0$, which proves nothing. \hfill$\square$
\end{example}

\subsection{Cube root: the canonical setting $R(t)=t^3-1$}

For the next three examples, $R(t)=t^3-1$. The roots are $\{1,\omega,\omega^2\}$, and the order-$3$ M\"obius transformation cyclically permuting them is simply $S(t) = \omega t$. The fixed points of $S$ are $\alpha=0,\beta=\infty$, so the substitution $z=(t-\alpha)/(t-\beta)$ degenerates to $z=t$ and no preliminary change of coordinates is required. The constants in~\eqref{eq:R-as-z3-K} are $c=1$, $K=1$, so the canonical radical is $\sqrt[3]{z^3-1}=\sqrt[3]{t^3-1}$.

In this setting, the eigendecomposition of $H(z)=F(z)$ under $z\mapsto\omega z$ is exactly the eigendecomposition of $F$ under multiplication by $\omega$.

\begin{example}[$F=1$: elementary]\label{ex:F1}
$H(z)=1\in V_0$, so $H_1\equiv 0$ and Theorem~\ref{thm:cubic-goursat} predicts elementarity. With $\varphi_0(x)=1$ and $K=1$:
\[
  \int \frac{\d t}{\sqrt[3]{t^3-1}} \;=\; J_0 \;=\; \frac{1}{3}\int \frac{\d x}{\sqrt[3]{x^2(x-1)}}.
\]
Apply the parametrization~\eqref{eq:Y0-param} of $Y_0\colon y^3=x^2(x-1)$:
\[
  x = \frac{1}{1-w^3},\qquad y = \frac{w}{1-w^3},\qquad \frac{\d x}{y} = \frac{3w}{1-w^3}\,\d w.
\]
Hence
\[
  J_0 \;=\; \frac{1}{3}\int \frac{\d x}{y} \;=\; \int \frac{w\,\d w}{1-w^3}.
\]
The remaining integrand has the partial-fraction decomposition
\[
  \frac{w}{1-w^3} \;=\; \frac{w}{(1-w)(1+w+w^2)} \;=\; \frac{A}{1-w} + \frac{Bw+C}{1+w+w^2}
\]
with $A=\tfrac{1}{3}$, $B=\tfrac{1}{3}$, $C=-\tfrac{1}{3}$, giving
\[
  \int \frac{w\,\d w}{1-w^3} \;=\; -\frac{1}{3}\log(1-w) + \frac{1}{6}\log(1+w+w^2) - \frac{1}{\sqrt{3}}\,\arctan\!\Bigl(\frac{2w+1}{\sqrt{3}}\Bigr) + C.
\]
Inverting the parametrization to recover $w$ in terms of $t$: from $x=t^3=1/(1-w^3)$ we get $1-w^3=1/t^3$, hence $w^3 = (t^3-1)/t^3$ and
\[
  w \;=\; \frac{\sqrt[3]{t^3-1}}{t}.
\]
Substituting back yields the explicit antiderivative:
\[
\begin{aligned}
  \int \frac{\d t}{\sqrt[3]{t^3-1}}
   = &-\tfrac{1}{3}\log\!\Bigl(1-\tfrac{\sqrt[3]{t^3-1}}{t}\Bigr)
     + \tfrac{1}{6}\log\!\Bigl(1+\tfrac{\sqrt[3]{t^3-1}}{t} + \tfrac{\sqrt[3]{(t^3-1)^2}}{t^2}\Bigr) \\
   & - \tfrac{1}{\sqrt{3}}\arctan\!\Bigl(\tfrac{1}{\sqrt{3}}\bigl(2\tfrac{\sqrt[3]{t^3-1}}{t} + 1\bigr)\Bigr) + C.
\end{aligned}
\]
One verifies this directly: with $w(t) = \sqrt[3]{t^3-1}/t$, a short calculation gives $\d w/\d t = 1/(t^2(t^3-1)^{2/3})$, while $w/(1-w^3) = w\cdot t^3 = t^2 \sqrt[3]{t^3-1}$, so that the chain rule yields $\d/\d t \int w\,\d w/(1-w^3) = w/(1-w^3)\cdot \d w/\d t = 1/\sqrt[3]{t^3-1}$, as required. \hfill$\square$
\end{example}

\begin{example}[$F=t^2$: trivially elementary]\label{ex:Ft2}
$H(z)=z^2\in V_2$, so $H_0=H_1=0$ and only $J_2$ contributes. With $\varphi_2(x)=1$:
\[
  J_2 \;=\; \frac{1}{3}\int \frac{\d x}{(x-1)^{1/3}}\bigg|_{x=t^3} \;=\; \frac{1}{3}\cdot\frac{3}{2}(x-1)^{2/3}\bigg|_{x=t^3} + C \;=\; \frac{1}{2}(t^3-1)^{2/3} + C.
\]
This recovers the obvious identity $\int t^2\,\d t/\sqrt[3]{t^3-1} = \tfrac{1}{2}\bigl(\sqrt[3]{t^3-1}\bigr)^2 + C$, which one would obtain by inspection (the integrand is, up to a factor of $3$, the derivative of $(t^3-1)^{2/3}$). The eigenvalue framework is consistent with --- and gives a different proof of --- this elementary observation. \hfill$\square$
\end{example}

\begin{example}[$F=t$: \emph{not} elementary]\label{ex:Ft}
$H(z)=z\in V_1$, so $H_0=H_2=0$ and only $J_1$ contributes. With $\varphi_1(x)=1$ and $K=1$:
\[
  \int \frac{t\,\d t}{\sqrt[3]{t^3-1}} \;=\; J_1 \;=\; \frac{1}{3}\int \frac{\d x}{\sqrt[3]{x(x-1)}}.
\]
Here $\omega_1=\d x/y$ on $Y_1\colon y^3=x(x-1)$ is of the second kind, and by~\eqref{eq:poly-criterion} it is not exact (the sum is $a_0=1\ne0$). By Theorem~\ref{thm:cubic-goursat}(iii) the integral is not elementary. Equivalently, working directly on $y^3=t^3-1$: $F=t$ is Hermite-reduced, and the coefficient of $t^{-1}$ in $t\,(t^3-1)^{-1/3}=1+\tfrac13t^{-3}+\cdots$ vanishes, so Proposition~\ref{prop:hermite} applies. This is also the case $m=1$, $n=3$, $p=-\tfrac13$ of Chebyshev's theorem.

The corresponding complete integral is a period of the CM curve $y^3 = x(x-1)$, of the kind governed by the Chowla--Selberg formula~\cite{ChowlaSelberg1967,Gross1978}. \hfill$\square$
\end{example}

The contrast among the three examples~\ref{ex:F1}--\ref{ex:Ft} is striking: $\int \d t/\sqrt[3]{t^3-1}$ is elementary (yet non-trivially so), $\int t^2\,\d t/\sqrt[3]{t^3-1}$ is elementary (and trivially), and $\int t\,\d t/\sqrt[3]{t^3-1}$ is not elementary.

\subsection{When $H_1\neq0$ and the integral is nevertheless elementary}

\begin{example}[An exact differential in $V_1$]\label{ex:exact-V1}
Take $R=t^3-1$ and the polynomial $F=t\,(2t^3-1)$. Then $H=H_1=F$ and $\varphi_1=2x-1$. Criterion~\eqref{eq:poly-criterion} reads $-1+2\cdot\tfrac12=0$, so $\omega_1$ is exact; indeed
\[
  \int\frac{t\,(2t^3-1)}{\sqrt[3]{t^3-1}}\,\d t\;=\;\tfrac12\,t^2\,(t^3-1)^{2/3}+C .
\]
Similarly $\int(2t^3-1)\,(t^3-1)^{-2/3}\,\d t=t\,(t^3-1)^{1/3}$, with $\widetilde H=\widetilde H_0\neq0$. \hfill$\square$
\end{example}

\begin{example}[$F=1$, $R=t^3+t$: a transposition]\label{ex:F1-symmetric}
Here $\mathcal H=1-3t^2$, so $\alpha,\beta=\pm1/\sqrt3$ are finite, the canonical form is $\tfrac{4\sqrt3}{9}(z^3+1)$, and $H=c_0/(1-z)$ has $H_1=c_0z/(1-z^3)\neq0$. Nevertheless the integral is elementary: with $s=t^{-2}$ and $v=(1+s)^{1/3}$,
\[
  \int\frac{\d t}{\sqrt[3]{t^3+t}}\;=\;-\frac32\int\frac{v\,\d v}{v^3-1} .
\]
The explanation is the involution $t\mapsto-t$, which fixes the root $0$ and exchanges $\pm i$. In the language of Proposition~\ref{prop:invariant}, $x_0=\mathcal G/R=-2(9t^2+1)/(t^3+t)$ changes sign under $t\mapsto-t$ while $F=1$ does not, so the coefficient $C(x_0)$ of~\eqref{eq:decomp13} is odd in $x_0$ and the whole $V_1$-part is integrated by~\eqref{eq:odd-part}; one finds
\[
  \int\frac{\d t}{\sqrt[3]{t^3+t}}\;=\;\sum_{\gamma=\pm1/\sqrt3}\int\frac{-u\,\d u}{u^3-1}\bigg|_{u=\frac{y}{t-\gamma}}\;+\;\int\frac{3w\,\d w}{2w^3+54}\bigg|_{w=\frac{(1-3t^2)\,y}{t^3+t}},\qquad y=\sqrt[3]{t^3+t}.
\]
Note that $F=1$ is a polynomial: polynomial $F$ does not make $\varphi_1$ a polynomial when $\beta$ is finite (Remark~\ref{rmk:f-poly-sufficient}). \hfill$\square$
\end{example}

\subsection{A combined integrand}

\begin{example}\label{ex:combined}
Take $R(t)=t^3-1$ and $F(t) = 1 + 5t^2 - 7\,t/(t^3+1)$. The three terms lie in $V_0$, $V_2$ and $V_1$ respectively, the last being $t\,\varphi_1(t^3)$ with $\varphi_1(x)=-7/(x+1)$. So $H_1\neq0$, and Theorem~\ref{thm:cubic-goursat}(i) does not apply to the third term. One might be tempted to conclude that
\[
  \int\frac{t\,\d t}{(t^3+1)\sqrt[3]{t^3-1}}\;=\;\frac13\int\frac{\d x}{(x+1)\sqrt[3]{x(x-1)}}
\]
is a non-elementary Abelian integral on $Y_1$. It is not. The integrand is invariant under $t\mapsto1/t$, which fixes the root $1$ and exchanges $\omega,\omega^2$; equivalently, with $\mathcal H=9t$ and $x_0=-27(t^3+1)/(t^3-1)$ we have $t/(t^3+1)=-\tfrac{3}{x_0}\cdot\tfrac{\mathcal H}{R}$, so $C(x_0)=-3/x_0$ is odd. By~\eqref{eq:odd-part} with $\psi(X)=-3/X$ and $\Delta=-27$,
\[
  \int\frac{t\,\d t}{(t^3+1)\sqrt[3]{t^3-1}}\;=\;-9\int\frac{w\,\d w}{4w^3+729},\qquad w=\frac{9t}{(t^3-1)^{2/3}} .
\]
(An equivalent form uses the invariant $v=\sqrt[3]{t^3-1}/(t-1)$ of the involution: the integral is $-3\int v(v^3-1)\,\d v/((4v^3-1)(v^3+2))$.) Hence the whole integral is elementary:
\[
  \int F\,\frac{\d t}{\sqrt[3]{t^3-1}} \;=\; \underbrace{\int \frac{\d t}{\sqrt[3]{t^3-1}}}_{\text{Example~\ref{ex:F1}}} \;+\; \frac52\,(t^3-1)^{2/3}\;+\;63\int\frac{w\,\d w}{4w^3+729} .
\]
\hfill$\square$
\end{example}

\subsection{A non-canonical $S$: shifted cube}

\begin{example}\label{ex:non-canonical-S}
Consider $R(t)=(t-1)(t-2)(t-3)$ and the order-$3$ M\"obius transformation $S$ with $S(1)=2$, $S(2)=3$, $S(3)=1$. Solving the three linear conditions gives
\[
  S(t) = \frac{5t-13}{3t-7}, \qquad \alpha,\beta = 2 \pm \frac{i}{\sqrt 3}.
\]
A direct check confirms $S(1)=-8/-4=2$, $S(2)=-3/-1=3$, $S(3)=2/2=1$, and the fixed points are the roots of $3t^2-12t+13=0$, which is the Hessian $\mathcal H$ of $R=t^3-6t^2+11t-6$ computed from~\eqref{eq:covariants}: no M\"obius transformation needs to be found.

By Proposition~\ref{prop:invariant}(b), $\int(at+b)\,\d t/(\mathcal H\sqrt[3]{R})$ is elementary. For $a=0$, $b=1$ the two components along $1/(t-\alpha)$ and $1/(t-\beta)$ give the same rational integrand, and
\[
  \int\frac{\d t}{(3t^2-12t+13)\,\sqrt[3]{(t-1)(t-2)(t-3)}}\;=\;\sum_{\gamma\in\{\alpha,\beta\}}\int\frac{-u\,\d u}{4u^3-2}\bigg|_{u=\frac{\sqrt[3]{R}}{t-\gamma}} .
\]
The two summands are complex conjugates of one another for real $t$.

For $F=1$ the coefficient $C$ in~\eqref{eq:decomp13} is non-zero but odd in $x_0$, reflecting the symmetry $t\mapsto4-t$ of $R$, and $\int\d t/\sqrt[3]{(t-1)(t-2)(t-3)}$ is again elementary. For a cubic without such a symmetry, e.g.\ $R=t^3+t+1$, the residual of $\int\d t/\sqrt[3]{R}$ is of the third kind and the methods of this paper do not decide it.
\hfill$\square$
\end{example}

\subsection{A second-level example}

\begin{example}\label{ex:second-level}
Consider
\[
  I\;=\;\int\frac{t^6-2t^3+4}{(t^3-1)(t^3-4)\,(t^3-2)^{2/3}}\,\d t .
\]
Here $R=t^3-2$, $\Delta=-108$, $\mathcal H=18t$, $\mathcal G=-54(t^3+2)$ and $x_0=-54(t^3+2)/(t^3-2)$. The integrand has $F\in\C(t^3)=\C(x_0)$, so in Proposition~\ref{prop:invariant}(c) we have $A=B=0$ and $F=C(x_0)$: the whole of $F$ lies in the space $\widetilde V_0$ on which Corollary~\ref{cor:two-thirds} gives no information. Moreover
\[
  F\;=\;-\frac{2\,(x_0^2+8748)}{x_0^2-26244}
\]
is \emph{even} in $x_0$, so the transpositions give nothing either, and after Hermite reduction the residual is of the third kind (simple poles at $t^3=1$ and $t^3=4$). At the first level the integral is undecided.

Now apply Proposition~\ref{prop:second-level}. With $w=18t/y^2$ and $Y^2=4w^3+2916=4(w^3+9^3)$ we get $\Phi(w)=3\chi(4w^3+2916)=-6\,(w^3+2916)/(w^3-5832)$, and putting $w=9v$,
\[
  I\;=\;\int\frac{(v^3+4)\,\d v}{(v^3-8)\sqrt{v^3+1}},\qquad v=\frac{2t}{(t^3-2)^{2/3}},\qquad \sqrt{v^3+1}=\frac{t^3+2}{t^3-2}.
\]
Up to the change $v\mapsto-v$ this is the first of Goursat's three integrals in Section~\ref{sec:order3}: it satisfies~\eqref{eq:V4-sum} on $Y^2=v^3+1$, so it is elementary by Theorem~\ref{thm:goursat2}. Explicitly, since $3(s^3-4)/(s^3+8)=\sum_{r^3=1}(s-r)/(s+2r)$ and each summand is integrated by Theorem~\ref{thm:order3} at the root $r$, one finds with $s=-v$, $\sigma=(t^3+2)/(t^3-2)$ and $\varepsilon_1=1$, $\varepsilon_\omega=-\omega^2$, $\varepsilon_{\omega^2}=-\omega$ (so that $\varepsilon_r^2=r$):
\[
  I\;=\;-\frac19\sum_{r^3=1}\frac1{\varepsilon_r}\,
  \log\frac{s^3-12\,r\,s^2-6\,r^2s-10+6\,\varepsilon_r\,(s-r)\,\sigma}{(s+2r)^3},
  \qquad s=-\frac{2t}{(t^3-2)^{2/3}} .
\]
The terms for $r=\omega,\omega^2$ are complex conjugates for real $t$. One checks the formula by differentiation; numerically, $\int_2^3=0.2586144917524889\ldots$ on both sides.

The substitution $t\mapsto v=2t/(t^3-2)^{2/3}$ is not M\"obius, and the involutions of $Y^2=v^3+1$ used here do not come from M\"obius transformations of $t$: this integral is out of reach of Theorem~\ref{thm:cubic-goursat} and of the transpositions, and is found only at the second level. \hfill$\square$
\end{example}

\subsection{A degenerate cubic: degree-$2$ radicand}\label{ex:cubic-degree-2}

\begin{example}
We illustrate Corollary~\ref{cor:degree2-nonelementary} concretely with $R(t) = t^2 - 1$, viewed as a cubic with one root at $\infty$ (cf.\ Remark~\ref{rmk:degree-2}). Here $r_1 = 1$, $r_2 = -1$, so the formula of Remark~\ref{rmk:degree-2}(b) gives
\[
  S(t) \;=\; \frac{1\cdot t - (1 + 1 + 1)}{t - (-1)} \;=\; \frac{t-3}{t+1},
\]
with fixed-point equation $t^2 + 3 = 0$ and so $\alpha = i\sqrt{3}$, $\beta = -i\sqrt{3}$. A direct check confirms $S(1) = -1$, $S(-1) = \infty$, $S(\infty) = 1$.

Substituting $z = (t-\alpha)/(t-\beta)$ with inverse $t(z) = i\sqrt{3}\,(1+z)/(1-z)$:
\[
  R(t(z))\,(1-z)^3 \;=\; \bigl[t(z)^2 - 1\bigr](1-z)^3 \;=\; 4(z^3 - 1),
\]
confirming $K = 1$ as predicted by Remark~\ref{rmk:degree-2}(a), with $c = 4$. The function~\eqref{eq:H-formula} for $F = 1$ is
\[
  H(z) \;=\; \frac{2i\sqrt{3}}{\sqrt[3]{4}\,(1-z)} \;=\; \frac{c_0}{1-z},\qquad c_0 := \frac{2i\sqrt{3}}{\sqrt[3]{4}}\neq 0.
\]
By Corollary~\ref{cor:degree2-nonelementary},
\[
  \int \frac{\d t}{\sqrt[3]{t^2-1}} \quad\text{is \emph{not} elementary.}
\]

In the framework of Theorem~\ref{thm:cubic-goursat} the residual is
\[
  J_1 \;=\; \frac{1}{3}\int \frac{\varphi_1(x)\,\d x}{\sqrt[3]{x(x-1)}},\qquad \varphi_1(x) \;=\; \frac{c_0}{1-x} ,
\]
whose only pole is at the branch point $x=K=1$; so $\omega_1$ is of the second kind, and it is not exact (its Hermite reduction on $Y_1$ is a non-zero multiple of $\d x/y$). The pieces $J_0$ and $J_2$ are elementary.

The contrast with the square-root case $\int\d t/\sqrt{t^2-1} = \operatorname{arccosh}(t) + C$ is geometric: the curve $y^2 = t^2-1$ has genus $0$ (it is a smooth conic), but $y^3 = t^2-1$ has genus $1$. \hfill$\square$
\end{example}

\section{Concluding remarks}
\label{sec:remarks}

\subsection*{Higher radicals.}
Part (i) of Theorem~\ref{thm:cubic-goursat} extends to $n$-th roots. Consider $\int F(t)\,\d t/\sqrt[n]{R(t)}$ with $\deg R=n$ and suppose that a M\"obius transformation $S$ of order $n$ permutes the roots of $R$ in a single cycle. (For $n=3$ this is automatic; for $n\ge4$ it is a condition on $R$.) After the substitutions $z = (t-\alpha)/(t-\beta)$ and $x = z^n$, the differential decomposes into $n$ eigencomponents under $z\mapsto\zeta z$ (with $\zeta = e^{2\pi i/n}$), the $k$-th of which lands on the curve
\[
  Y_k\colon\; y^n \;=\; x^{n-1-k}(x-K), \qquad k = 0, 1, \dots, n-1.
\]
Riemann--Hurwitz for the projection $Y_k\to\mathbb{P}^1$, $(x,y)\mapsto x$, gives
\begin{equation}\label{eq:gen-n-genus}
  g(Y_k) \;=\; \frac{n+1\;-\;\gcd(n,k)\;-\;\gcd(n,k+1)}{2}.
\end{equation}
Since $\gcd(n,k)$ and $\gcd(n,k+1)$ are coprime divisors of $n$, their sum equals $n+1$ if and only if one of them equals $n$, which occurs precisely for $k = 0$ or $k = n-1$. Therefore
\[
  g(Y_k) = 0 \;\iff\; k\in\{0, n-1\},
\]
and every interior $k\in\{1,\dots,n-2\}$ gives a curve of positive genus. The two outermost eigencomponents always integrate elementarily --- via $x = K/(1-w^n)$ and $x = K + w^n$ --- while for the interior components the symmetry gives no information.

\smallskip\noindent
$\bullet$ \emph{$n=2$.} There is no interior $k$, but this is the trivial case of a quadratic radicand, of genus $0$ from the start. It is \emph{not} Goursat's theorem, where $R$ is quartic and its roots form two orbits of the involution; there one of the two eigencomponents descends to genus $1$ (Section~\ref{sec:geometric}).

\smallskip\noindent
$\bullet$ \emph{$n=3$.} One interior component, on $Y_1\colon y^3 = x(x-K)$, of genus $1$ and $j$-invariant $0$.

\smallskip\noindent
$\bullet$ \emph{$n=4$.} Two interior components, both of genus $1$. Each $Y_k$ carries the automorphism $y\mapsto iy$ of order $4$, which fixes the point over $x=K$; hence $j=1728$ for both, for every $K$. (Explicitly, $u=y^2/x$ maps $Y_1\colon y^4=x^2(x-K)$ to $y^2=u^3+Ku$.) The condition on $R$ is that its roots be harmonic.

\smallskip\noindent
$\bullet$ \emph{$n=5$.} Three interior components, all of genus $2$. The curves $y^5=x^a(x-K)^b$ in question have exponents $(a,b,c)$ at $(0,K,\infty)$ equal to $(3,1,1)$, $(2,1,2)$, $(1,1,3)$ modulo $5$; these triples are permutations of one another up to multiplication by the unit $2$, so the three curves are isomorphic, and isomorphic to $y^2=x^5-1$, the unique curve of genus $2$ with an automorphism of order $5$.

\smallskip\noindent
In every case the rigidity is the same phenomenon: the interior curves inherit the automorphism $y\mapsto\zeta y$. It is not special to $n=4$.

\subsection*{What remains undecided.}
The methods of this paper decide an integral $\int F\,\d t/R^{p}$ on a genus-$1$ curve in two situations: when the symmetry reductions exhaust it, and when what they leave is, after Hermite reduction, of the second kind. If the residual is of the third kind one needs the divisor-class (torsion) test of the general theory. For example, for a cubic $R$ without special symmetry, $\int\d t/\sqrt[3]{R}$ has residues $\propto(1,\omega,\omega^2)$ at the three points over $\infty$, and it is elementary only if the differences of these points are torsion on $X_R$; this holds for $R=t^3-1$, where they are flexes.

\subsection*{Algorithmic perspective.}
The Risch--Trager--Bronstein algorithm~\cite{Risch1969,Risch1970,Davenport1981,Trager1984,Bronstein1990,Bronstein2005} provides, in principle, a uniform decision procedure for elementary integration over algebraic function fields, including cube-radical integrands; in this sense it subsumes both Goursat's theorem and our cube-root analog. In practice, the algebraic case is complex and has not been implemented in fully general form in any computer algebra system (see also the discussion in~\cite{Blake2020}); even the most complete implementation, in FriCAS, can be slow on inputs of modest size and fails on a number of integrands handled easily by classical methods --- including Goursat-type pseudo-elliptics noted on \texttt{sci.math.symbolic}~\cite{WelzNewsgroup}. The reductions of this paper are cheap, preserve and explain the structure, and produce compact closed forms; combined with Hermite reduction they give a useful partial decision procedure. They are not a replacement for the general algorithm.

\subsection*{Differential Galois theory.}
At a higher level, the question of which integrals are elementary is governed by the differential Galois group of the relevant Picard--Vessiot extension; we refer to the standard texts of Magid~\cite{Magid1994}, van der Put--Singer~\cite{vanderPutSinger2003}, and Crespo--Hajto~\cite{CrespoHajto2011} for the algebraic theory, and to Khovanskii~\cite{Khovanskii2014} for a topological viewpoint. For curves of positive genus, the obstructions to elementarity are governed by the Jacobian and the divisor class group~\cite{Bertrand1990}. Goursat's theorem and its cube-root analog are explicit computations of the part of these obstructions that is forced to vanish by automorphisms of the curve.

\appendix
\section{Implementation}\label{sec:appendix}

The procedure described in this paper has been implemented in rational arithmetic, following Propositions~\ref{prop:hermite}, \ref{prop:rootfree-sqrt} and~\ref{prop:invariant} and Theorem~\ref{thm:order3}. For input $F\in\Q(t)$, squarefree $R\in\Q[t]$ and $p\in\{\tfrac12,\tfrac13,\tfrac23\}$ (with $\deg R\in\{3,4\}$ for $p=\tfrac12$ and $\deg R\in\{2,3\}$ otherwise) it proceeds as follows.

\begin{enumerate}
\item \emph{Hermite reduction} of $F$: $\int F\,\d t/R^p=A\,R^{1-p}+\int\widetilde F\,\d t/R^p$. If $\widetilde F=0$ the integral is algebraic.
\item \emph{Symmetry reduction} of $\widetilde F$.
  \begin{itemize}
  \item $p=\tfrac12$: if $I(R)=0$, the component $\varphi_1(X)z$ of Theorem~\ref{thm:order3} is removed; then, for each involution $S_\theta$ in turn, the anti-invariant part of what remains is integrated by Proposition~\ref{prop:rootfree-sqrt}(c),(d); finally the order-$3$ step is applied once more to the $V_4$-invariant remainder.
  \item $p=\tfrac13,\tfrac23$: the decomposition of Proposition~\ref{prop:invariant}(b) or (c) is computed over $\Q(x_0)$; the two components along the roots of $\mathcal H$ are integrated through $u=y/(t-\gamma)$, and the odd part of $C$ through $w$ by~\eqref{eq:odd-part}.
  \end{itemize}
  Each step produces a rational integrand in a parameter that is a rational function of $t$ and $y=R^{1/n}$.
\item \emph{Hermite reduction of the residual}, and classification by Proposition~\ref{prop:hermite}.
\item For $p=\tfrac13,\tfrac23$, if the residual is of the third kind: \emph{second level}. The residual is rewritten as $\int\Phi(w)\,\d w/\sqrt{4w^3-27\Delta}$ by Proposition~\ref{prop:second-level} and steps 1--3 are applied to it with $p=\tfrac12$.
\end{enumerate}
The output is an antiderivative of the reduced part, the reduced residual $\widetilde F_{\mathrm{res}}$, and one of three verdicts:
\begin{center}
\begin{tabular}{ll}
\texttt{elementary} & $\widetilde F_{\mathrm{res}}=0$;\\
\texttt{non-elementary} & $\widetilde F_{\mathrm{res}}\neq0$ and $\widetilde F_{\mathrm{res}}\,\d t/R^p$ is of the second kind (a proof);\\
\texttt{undecided} & $\widetilde F_{\mathrm{res}}\,\d t/R^p$ is of the third kind.
\end{tabular}
\end{center}
A failure of the symmetry criterion alone is never reported as non-elementarity.

Three features of the implementation are worth recording.
\begin{itemize}
\item \emph{No roots of $R$.} For $p=\tfrac12$ the only algebraic numbers are the roots of the resolvent cubic (or of $R$, if $R$ is cubic), and Goursat's criterion itself is evaluated over $\Q$ as a trace. For $p=\tfrac13,\tfrac23$ the only algebraic numbers are the two roots of $\mathcal H$, in $\Q(\sqrt{-3\Delta})$. No roots of unity, cube roots of constants or fixed-point coordinates occur.
\item \emph{Descent by remainders.} A function $g(t)$ known to be a function of $X(t)=N/D$ is written as $G(X)$ by reducing its numerator and denominator modulo $N-xD$ in $\Q(x)[t]$ and taking the ratio of corresponding coefficients.
\item \emph{Branch-free output.} The antiderivative is expressed in $t$ and the single radical $y=R^{1/n}$, and is verified by differentiating formally with $y'=R'y/(nR)$.
\end{itemize}

\subsection*{Sample results}

\begin{center}
\small
\begin{tabular}{lll}
integrand & verdict & antiderivative / residual\\\hline
$t/\sqrt{(t^2-1)(t^2-4)}$ & elementary & $\tfrac12\log\frac{m+1}{m-1}$, $m=y/(t^2-4)$\\[2pt]
$(3t^4-10t^2+4)/\sqrt{(t^2-1)(t^2-4)}$ & elementary & $t\,y$\\
$t/\bigl((t^3+8)\sqrt{t^3-1}\bigr)$ & elementary & (Legendre--Clausen)\\
$1/\bigl(t\sqrt{t^3-1}\bigr)$ & elementary & $\tfrac23\arctan y$\\
$t/\sqrt{t^3+1}$, \ $1/\sqrt{t^4+1}$ & non-elementary & residual $t$, resp.\ $1$\\
$1/\sqrt[3]{t^3-1}$ & elementary & Example~\ref{ex:F1}\\
$t/\sqrt[3]{t^3-1}$, \ $1/\sqrt[3]{t^2-1}$ & non-elementary & residual $t$, resp.\ $1$\\
$t(2t^3-1)/\sqrt[3]{t^3-1}$ & elementary & $\tfrac12t^2y^2$\\
$1/\sqrt[3]{t^3+t}$ & elementary & Example~\ref{ex:F1-symmetric}\\
$1+5t^2-7t/(t^3+1)$ over $\sqrt[3]{t^3-1}$ & elementary & Example~\ref{ex:combined}\\
$1/\sqrt[3]{t^3+t+1}$ & undecided & residual of the third kind\\
$(t^6-2t^3+4)/\bigl((t^3-1)(t^3-4)(t^3-2)^{2/3}\bigr)$ & elementary & Example~\ref{ex:second-level} (second level)\\
$t/(t^3-1)^{2/3}$ & elementary & $-\int s\,\d s/(s^3-1)$, $s=t/y$\\
$1/(t^3-1)^{2/3}$ & non-elementary & residual $1$ (first kind)\\
\end{tabular}
\end{center}

As an independent check, Chebyshev's theorem gives the correct answer for $\int t^m(t^n-1)^{-p}\,\d t$; for $-5\le m\le7$ and $(n,p)\in\{(3,\tfrac13),(3,\tfrac23),(2,\tfrac13),(2,\tfrac23),(3,\tfrac12),(4,\tfrac12)\}$ the implementation decides all $78$ integrals, in agreement with it.

\subsection*{Limitations}
(1) Residuals of the third kind are not decided. (2) When the resolvent cubic is irreducible, Goursat's criterion is decided over $\Q$ but the explicit reduction requires arithmetic in a cubic number field. (3) For a quartic $R$ with $I(R)=0$ the order-$3$ step is carried out only when $R$ has a rational root. (4) Logarithmic terms are returned as sums over the roots of a polynomial; conjugate complex terms are not recombined into real form.


\begin{thebibliography}{99}

\bibitem{Bertrand1990}
D.~Bertrand,
\emph{Extensions de $\mathcal{D}$-modules et groupes de Galois diff\'erentiels},
in: \emph{p-adic analysis (Trento, 1989)}, Lecture Notes in Math., vol.~1454, Springer, 1990, pp.~125--141.

\bibitem{Blake2020}
S.~Blake,
\emph{A simple method for computing some pseudo-elliptic integrals in terms of elementary functions},
preprint, 2020. arXiv:2004.04910.

\bibitem{Bronstein1990}
M.~Bronstein,
\emph{Integration of elementary functions},
J.~Symbolic Comput.\ \textbf{9} (1990), 117--173.

\bibitem{Bronstein2005}
M.~Bronstein,
\emph{Symbolic Integration I: Transcendental Functions},
2nd ed., Algorithms and Computation in Mathematics, vol.~1, Springer, 2005.


\bibitem{Chebyshev1853}
P.~L.~Chebyshev (Tchebichef),
\emph{Sur l'int\'egration des diff\'erentielles irrationnelles},
J.~Math.\ Pures Appl.\ \textbf{18} (1853), 87--111.


\bibitem{ChowlaSelberg1967}
S.~Chowla and A.~Selberg,
\emph{On Epstein's zeta function},
J.~reine angew.\ Math.\ \textbf{227} (1967), 86--110.

\bibitem{CrespoHajto2011}
T.~Crespo and Z.~Hajto,
\emph{Algebraic Groups and Differential Galois Theory},
Graduate Studies in Mathematics, vol.~122, American Mathematical Society, 2011.

\bibitem{Davenport1981}
J.~H.~Davenport,
\emph{On the Integration of Algebraic Functions},
Lecture Notes in Computer Science, vol.~102, Springer, 1981.

\bibitem{Goursat1887}
\'E.~Goursat,
\emph{Note sur quelques int\'egrales pseudo-elliptiques},
Bulletin de la S.~M.~F., \textbf{15} (1887), 106--120.

\bibitem{Gross1978}
B.~H.~Gross,
\emph{On the periods of abelian integrals and a formula of Chowla and Selberg} (with an appendix by D.~E.~Rohrlich),
Invent.\ Math.\ \textbf{45} (1978), 193--211.

\bibitem{Khovanskii2014}
A.~Khovanskii,
\emph{Topological Galois Theory: Solvability and Unsolvability of Equations in Finite Terms},
Springer Monographs in Mathematics, Springer, 2014.

\bibitem{Magid1994}
A.~R.~Magid,
\emph{Lectures on Differential Galois Theory},
University Lecture Series, vol.~7, American Mathematical Society, 1994.

\bibitem{Risch1969}
R.~H.~Risch,
\emph{The problem of integration in finite terms},
Trans.\ Amer.\ Math.\ Soc.\ \textbf{139} (1969), 167--189.

\bibitem{Risch1970}
R.~H.~Risch,
\emph{The solution of the problem of integration in finite terms},
Bull.\ Amer.\ Math.\ Soc.\ \textbf{76} (1970), 605--608.

\bibitem{Ritt1948}
J.~F.~Ritt,
\emph{Integration in Finite Terms},
Columbia University Press, 1948.

\bibitem{Rosenlicht1972}
M.~Rosenlicht,
\emph{Integration in finite terms},
Amer.\ Math.\ Monthly \textbf{79} (1972), 963--972.

\bibitem{Trager1984}
B.~M.~Trager,
\emph{Integration of algebraic functions},
Ph.D.\ thesis, MIT, 1984.

\bibitem{vanderPutSinger2003}
M.~van der Put and M.~F.~Singer,
\emph{Galois Theory of Linear Differential Equations},
Grundlehren der mathematischen Wissenschaften, vol.~328, Springer, 2003.

\bibitem{WelzNewsgroup}
M.~Welz,
\emph{Postings on Goursat-type pseudo-elliptic integrals and their cube-root analogs}, sci.math.symbolic newsgroup, 2012--2024.\\
\url{https://groups.google.com/g/sci.math.symbolic/search?q=goursat}.

\end{thebibliography}
\end{document}